\documentclass[sigconf]{acmart}

\AtBeginDocument{%
  }

\copyrightyear{2025}
\acmYear{2025}
\setcopyright{cc}
\setcctype{by}
\acmConference[HSCC '25]{28th ACM International Conference on Hybrid Systems: Computation and Control}{May 6--9, 2025}{Irvine, CA, USA}
\acmBooktitle{28th ACM International Conference on Hybrid Systems: Computation and Control (HSCC '25), May 6--9, 2025, Irvine, CA, USA}
\acmDOI{10.1145/3716863.3718031}
\acmISBN{979-8-4007-1504-4/2025/05}

\usepackage{graphicx}      %

\makeatletter
\let\old@ssect\@ssect %
\makeatother
\usepackage{natbib}        %

\usepackage{amsmath}
\usepackage{amsfonts}

\AtEndPreamble{%
\newtheorem{remark}[theorem]{Remark}
\theoremstyle{acmdefinition}
\newtheorem{problem}[theorem]{Problem}
}

\usepackage{caption}
\usepackage{subcaption}
\usepackage{tabularx}

\usepackage{booktabs}
\usepackage{enumitem}
\usepackage{csquotes}

\usepackage{xcolor}

\usepackage{xspace}
\usepackage[acronym]{glossaries}

\usepackage{etoolbox}
\newbool{inproceedings}
\boolfalse{inproceedings}

\input{macros.tex}
\newacronym[longplural=Interval Markov Decision Processes]{imdp}{IMDP}{Interval Markov Decision Process}
\newacronym[longplural=Markov Decision Processes]{mdp}{MDP}{Markov Decision Process}
\newacronym{imc}{IMC}{Interval Markov Chain}
\newacronym[longplural=orthogonally decoupled Interval Markov Decision Processes]{odimdp}{odIMDP}{orthogonally decoupled Interval Markov Decision Process}
\newacronym{odimc}{odIMC}{orthogonally decoupled Interval Markov Chain}
\newacronym[longplural=Robust Markov Decision Processes]{rmdp}{RMDP}{Robust Markov Decision Process}
\newacronym{nndm}{NNDM}{Neural Network Dynamic Model}
\newacronym[longplural=Gaussian Processes]{gp}{GP}{Gaussian Process}
\newacronym{wrt}{wrt.}{with respect to}
\newacronym{pdf}{PDF}{probability density function}
\newacronym{iid}{i.i.d.}{independent and identically distributed}
\newacronym{sbf}{SBF}{Stochastic Barrier Function}
\newacronym{mc}{MC}{Markov Chain}

\begin{document}

\title[Scalable control synthesis for stochastic systems via structural IMDP abstractions]{Scalable control synthesis for stochastic systems\\via structural IMDP abstractions}

\author{Frederik Baymler Mathiesen}
\email{f.b.mathiesen@tudelft.nl}
\affiliation{%
  \institution{Delft University of Technology}
  \city{Delft}
  \country{The Netherlands}
}

\author{Sofie Haesaert}
\email{s.haesaert@tue.nl}
\affiliation{%
  \institution{Eindhoven University of Technology}
  \city{Eindhoven}
  \country{The Netherlands}
}

\author{Luca Laurenti}
\authornote{L.L. is partially supported by the NWO (Grant OCENW.M.22.056).}
\email{l.laurenti@tudelft.nl}
\affiliation{%
  \institution{Delft University of Technology}
  \city{Delft}
  \country{The Netherlands}
}




\begin{abstract}                %
This paper introduces a novel abstraction-based framework for controller synthesis of nonlinear discrete-time stochastic systems. The focus is on probabilistic reach-avoid specifications. The framework is based on abstracting a stochastic system into a new class of robust Markov models, called \acrfullpl{odimdp}. Specifically, an \acrshortpl{odimdp} is a class of robust Markov processes, where the transition probabilities between each pair of states are uncertain and have the product form. We show that such a specific form in the transition probabilities allows one to build compositional abstractions of stochastic systems that, for each state, are only required to store the marginal probability bounds of the original system. This leads to improved memory complexity for our approach compared to commonly employed abstraction-based approaches. Furthermore, we show that an optimal control strategy for a \acrshortpl{odimdp} can be computed by solving a set of linear problems. When the resulting strategy is mapped back to the original system, it is guaranteed to lead to reduced conservatism compared to existing approaches. To test our theoretical framework, we perform an extensive empirical comparison of our methods against \acrlong{imdp}- and \acrlong{mdp}-based approaches on various benchmarks including 7D systems. Our empirical analysis shows that our approach substantially outperforms state-of-the-art approaches in terms of both memory requirements and the conservatism of the results. 
\end{abstract}

\maketitle

\section{Introduction}
\label{sec:introdution}
Modern cyber-physical systems, such as autonomous vehicles or robotic systems, are often modeled as stochastic dynamical systems \cite{carvalho2014stochastic, 6907293}, where nonlinear dynamics transform stochastic quantities over time. Consequently, to deploy cyber-physical systems in safety-critical applications, it has become of paramount importance to develop techniques for controller synthesis and verification of nonlinear stochastic dynamical systems with guarantees of performance and correctness \cite{10.1007/978-3-031-06773-0_1}. However, unfortunately, existing techniques still lack scalability and are commonly limited to low-dimensional systems \cite{lavaei2022automated,laurenti2023unifying}. 

Abstraction-based methods arguably represent the state-of-the-art for formal analysis of nonlinear stochastic systems \cite{lavaei2022automated,liu2014abstraction, van2023syscore, cauchi2019efficiency}. Abstraction-based methods approximate a given complex system (the \enquote{concrete} system) with a simpler one (the \enquote{abstraction}) for the purpose of performing analysis while maintaining a (bi)simulation relation between the concrete and abstract systems \cite{liu2014abstraction, VANHUIJGEVOORT2022110476, 9993724, cauchi2019efficiency}. Then, by relying on the simulation relation, controller strategies or verification results can be mapped from the abstraction to the concrete system with provable correctness guarantees. In the context of stochastic dynamical systems, the abstraction model is usually taken to be a variant of a Markov model with finite state space.
In particular, abstractions to robust Markov models and, especially, \glspl{imdp} \cite{givan2000bounded} have recently shown great potential to achieve state-of-the-art results \cite{cauchi2019efficiency, jiang2022safe, chatterjee2008model, badings2023robust}.
However, these approaches are still limited both in terms of memory and computational requirements \cite{cauchi2019efficiency, laurenti2023unifying, mathiesen2024intervalmdpjlacceleratedvalueiteration}. %

In this paper, we introduce a novel framework for formal analysis of stochastic dynamical systems based on abstraction to robust Markov models. Our framework uses compositional reasoning and relies on the structural properties of the system, namely independence of the stochasticity between dimensions, to improve the tightness and computational efficiency of abstraction-based methods that use robust finite-state Markov models. In particular, we introduce a new class of robust Markov models, called \glspl{odimdp}, where the transition probabilities between each pair of states lie within a product of intervals. We show that this form of transition probability is particularly suitable to abstract stochastic systems, where each interval in the product can encode the marginal probabilities of the system. This approach has two main advantages: (i) for each pair of state actions, \glspl{odimdp} are only required to store the marginal distributions of the system instead of the transition probabilities to all the other states in the abstraction. This leads to substantial improvements in the memory requirement of the method, which is arguably the main bottleneck of existing abstraction-based methods, (ii) by encoding constraints on the marginals, the resulting abstraction model exploits the structural properties of the concrete system, leading to improved tightness of the resulting error bounds. 
To guarantee computational efficiency of our framework, we propose a novel approach to synthesize robust strategies for an \gls{odimdp} against a probabilistic reach-avoid property based on a divide-and-conquer algorithm that recursively solves a linear programming problem and returns upper and lower bounds on the probability that the odIMDP satisfies the property. Finally, we show that the resulting strategy can be mapped back on the concrete system and that the returned bounds are guaranteed to contain the probability that the concrete system satisfies the property.

To empirically evaluate the framework, we tested it on benchmarks from the literature \cite{b8abfba04ce648ef8322e66658774906, van2023syscore, 10.1007/978-3-031-68416-6_15, reed2023promises}, including stochastic affine and nonlinear systems, neural network dynamic models, \glspl{gp}, and Gaussian mixture models.
We compared our framework with existing state-of-the-art abstraction-based tools for verification of stochastic systems based on abstractions to \glspl{mdp} and \glspl{imdp} \cite{van2023syscore, 10.1007/978-3-031-68416-6_15}.
As expected, the experiments show that our approach outperforms the state-of-the-art both in terms of memory requirement and conservativity of the results, allowing one to synthesize optimal strategies for systems that current methodologies cannot handle. For example, on both 6D and 7D systems, while existing approaches returned trivial results, our approach would successfully synthesize optimal policies with a guaranteed lower bound on the satisfaction probability $>0.95$.
In summary, the main contributions of the paper are:
\begin{itemize}
    \item We introduce a new class of robust Markov models called \acrfullpl{odimdp} and show how robust value iteration for \glspl{odimdp} can be performed by relying on linear programming.
    \item We show that a large class of stochastic dynamical systems can be successfully abstracted as \glspl{odimdp} and derive efficient and scalable methods to construct the abstraction. 
    \item We performed a large-scale empirical evaluation and compared our approach with state-of-the-art abstraction-based verification methods for stochastic systems.
\end{itemize}

\paragraph{Related Works}
Abstraction-based approaches for stochastic systems generally rely on abstracting the concrete system to either \glspl{mdp} or \glspl{imdp} \cite{lavaei2022automated}. 
\glspl{odimdp} shares many features with \glspl{imdp} \cite{givan2000bounded} and abstraction of stochastic systems to \glspl{imdp} has been extensively studied, including stochastic hybrid systems \cite{cauchi2019efficiency}, neural networks \cite{adams2022formal}, and \acrlongpl{gp} \cite{jackson2021formal, reed2023promises}, and using data-driven construction methods \cite{jackson2020towards, badings2023robust}. Furthermore, many tools support value iteration over \glspl{imdp}, including: PRISM \cite{kwiatkowska2011prism}, bmdp-tool \cite{morteza_lahijanian_bmdp_tool}, Storm \cite{Hensel2021}, IntervalMDP.jl \cite{mathiesen2024intervalmdpjlacceleratedvalueiteration}, and IMPaCT\footnote{IMPaCT includes functionality for the abstraction of stochastic systems to \glspl{imdp}.} \cite{10.1007/978-3-031-68416-6_15}. However, a common limitation of these approaches is in the memory requirements in building the abstraction, which is generally quadratic in $|S|$, the size of the \gls{imdp}, which is, in turn, exponential in $n$, the dimension of the state space of the concrete system. In contrast, our approach based on \glspl{odimdp} guarantees an improved memory complexity (see Subsection \ref{subsec:space_complexity_analysis}) of a factor $\sqrt[n]{|S|}$. Furthermore, we also show how our approach generally returns tighter bounds compared to \gls{imdp}-based approaches. To the best of our knowledge, the only paper that has considered compositional reasoning for \glspl{imdp} \cite{10.1007/978-3-319-69483-2_2}, which, however, focuses on the parallel composition of given \glspl{imdp}, while in this paper, we focus on compositional construction of the abstraction. %
Approaches based on abstractions to (non-robust) \glspl{mdp} have also been widely employed in the literature \cite{lavaei2022automated,abate2010approximate,soudjani2014aggregation}. To scale to larger-dimensional systems, various works \cite{nilsson2018toward, van2023syscore, BOUTILIER200049, EsmaeilZadehSoudjani2016, LAVAEI2021100991} have devised frameworks that precede a gridding approach, such as the one introduced in our paper, with model order reductions or other modular approaches. However, in these approaches, gridding is still the limiting factor. Consequently, our approach is complementary to these order-reducing approaches and should lead, when combined with these methods, to abstractions of even higher-dimensional systems.

In parallel to abstraction-based methods, various works have focused on abstraction-free methods based on \glspl{sbf} \cite{4287147, SANTOYO2021109439}.
The idea is to construct a function that allows one to upper-bound the probability that the system exits a given safe set without directly studying the evolution of the system over time.
The synthesis of \glspl{sbf} has been formulated as sum-of-squares optimization \cite{SANTOYO2021109439, mazouz2022safety}, with neural networks \cite{9990576}, and as linear programming \cite{mazouz2024piecewisestochasticbarrierfunctions}.
Furthermore, data-driven synthesis techniques have also been developed \cite{SALAMATI20217, 10383306}.
However, compared to abstraction-based methods, \glspl{sbf} tend to be conservative \cite{laurenti2023unifying}.

\paragraph*{Structure of the paper}
In Section \ref{sec:problem_statement},  we formally state our problem of verifying probabilistic reach-avoid for stochastic systems and outline our approach.
Then in Section \ref{sec:interval_markov_decision_processes}, we briefly review \glspl{imdp} and how stochastic systems are abstracted to \glspl{imdp}. We also review weaknesses of this approach that leads us to introduce our new Markov model, \acrfullpl{odimdp}, which is introduced in Section \ref{sec:orthogonally_decoupled_interval_markov_decision_processes}. Section \ref{sec:orthogonally_decoupled_interval_markov_decision_processes} also discusses the space complexity required to store an \glspl{odimdp} and analyse their ambiguity set and compare them to those of state-of-the-art abstraction methods.
In Section \ref{sec:abstraction}, we show how a stochastic system can be successfully abstracted to \glspl{odimdp}. %
Then, in Section \ref{sec:value_iteration} we derive algorithms for robust value iteration of \glspl{odimdp} based on linear programming.
Finally, in Section \ref{sec:experiments}, we empirically benchmark the proposed methodology.

\paragraph*{Notation}
$\mathbb{R}$ and $\mathbb{N}$ represent the set of real and natural numbers with $\mathbb{N}_0 = \mathbb{N} \cup \{0\}$. The set $\mathbb{S}^n_{++}$ is the set of symmetric positive-definite matrices of size $n \in \mathbb{N}$. For a finite set $S$, $|S|$ is its cardinality. A probability distribution $\probdist$ over a finite set $S$ is a function $\probdist : S \to [0, 1]$ satisfying $\sum_{s \in S} \probdist(s) = 1$. $\dists(S)$ is the set of all probability distributions over $S$. 
For $\underline{p}, \overline{p} : S \to [0, 1]$ such that $\underline{p}(s) \leq \overline{p}(s)$ for each $s \in S$ and $\sum_{s \in S} \underline{p}(s) \leq 1 \leq \sum_{s \in S} \overline{p}(s)$, an interval ambiguity set $\ambiguityset \subset \dists(S)$ is the set of distributions such that $
    \ambiguityset = \{ \probdist \in \dists(S) \,:\, \underline{p}(s) \leq \probdist(s) \leq \overline{p}(s) \text{ for each } s\in S \}.
$
$\underline{p}, \overline{p}$ are often referred to as the interval bounds of the interval ambiguity set.
The set of all interval ambiguity sets over $S$ is denoted by $\intamb(S)$. For $n$ finite sets $S_1, \ldots, S_n$ we denote by $S_1 \times \cdots \times S_n$ their Cartesian product. Given $S=S_1 \times \cdots \times S_n$ and $n$ ambiguity sets $\ambiguityset_i \in \dists(S_i)$, $i = 1, \ldots, n$, the product ambiguity set $\ambiguityset \subseteq \dists(S)$ is defined as
 $   \ambiguityset = \left\{ \probdist \in \dists(S) \,:\, \probdist(s) = \prod_{i=1}^n \probdist^i(s^i), \, \probdist^i \in \ambiguityset_i \right\}$, 
where $s = (s_1, \ldots, s_n)\in S$. In what follows, with a slight abuse of notation, we denote the product ambiguity set as $\ambiguityset = \bigotimes_{i=1}^n \ambiguityset_i$. Each $\ambiguityset_i$ is called a marginal or component ambiguity set.
$\mathbf{1}_{X}(x)$ denotes the indicator function for the set $X$, that is, $\mathbf{1}_{X}(x) = 1$ if $x \in X$ and $0$ otherwise.

\section{Problem Statement}\label{sec:problem_statement}
Consider a stochastic process described by the following stochastic difference equation
\begin{equation}
    \label{eq:System}
    x_{k+1} = f(x_k, u_k, v_k)
\end{equation}
where $x_k \in \mathbb{R}^{n}$ and  $u_k \in U \subseteq \mathbb{R}^{m}$ are respectively state of the system and input at time $k$. We assume that $|U|$ is finite. $v_k \in V \subseteq \mathbb{R}^{n_v}$ is an \gls{iid} random variable for all time $k \in \mathbb{N}_{0}$ that represents the noise that affects the system at each time step. The vector field $f: \mathbb{R}^{n} \times U \times V \to \mathbb{R}^{n}$ represents the controlled stochastic dynamics of the system.

To define a probability measure for System~\eqref{eq:System}, we introduce the one-step stochastic kernel $T(X\mid x, u)$ that defines the probability that System~\eqref{eq:System} from state $x$ will transition to region $X\subseteq \mathbb{R}^n$ in one time step under action $u$. 
For this work, we assume that the transition kernel is a mixture of $K$ Gaussians for $K>0$, each with diagonal covariance. That is,
\begin{equation}
\label{eq:transition_kernel}
    T(X\mid x, u):= \sum_{r=1}^K \alpha_r(x, u) \int_{X} \mathcal{N}(x' \mid \mu^r(x,u), \Sigma^r(x, u)) dx',
\end{equation}
where $\mu^r(x, u) \in \mathbb{R}^{n}$ and $\Sigma^r(x, u) \in \mathbb{S}^{n}_{++}$ are action and state-dependent mean and covariance functions of the $r$-th Gaussian in the mixture. $\alpha_r(x, u)$ is a weighting function such that $\alpha_r(x, u) \geq 0$ and $\sum_{r=1}^K \alpha_r(x, u) = 1$ for each source-action pair $(x, u)$. We assume $\mu^r, \Sigma^r,$  and $\alpha_r$ are continuous in $x$ for each $u$.
When $K = 1$, we omit $r$ from the weighting probability function, mean, and covariance.

\begin{remark}
Note that, while the assumption of diagonal covariance for each distribution in  Eqn. \eqref{eq:transition_kernel} introduces independence of the stochasticity at different dimensions, the resulting model still encompasses a large class of models of practical interest. In fact, Eqn. \eqref{eq:transition_kernel} includes linear and nonlinear systems with additive Gaussian noise \cite{cauchi2019efficiency, adams2022formal, adams2024GMM} and systems learned via Gaussian process regression \cite{williams1995gaussian}. Furthermore, we should stress that mixtures of Gaussians with diagonal covariance can arbitrarily well approximate any distribution \cite[Proposition 4.1]{NESTORIDIS20111783}. %
\end{remark}

Given an initial condition $x_0 \in \mathbb{R}^{n}$ and a control policy\footnote{For the class of systems and problem considered in this paper time-dependent Markov policies suffices for optimality \cite{laurenti2023unifying}.} $\pi_x: \mathbb{R}^n \times \mathbb{N}_0 \to U$, the stochastic kernel $T$ induces a unique and well-defined probability measure $P^{x_0,\pi_x}$ \cite[Prop. 7.45]{bertsekas2004stochastic}, such that for sets
$X_0, X_{k + 1} \subseteq X$ it holds that 
\begin{align*}
    &P^{x_0, \pi_x}[x_0\in X_0] =  \mathbf{1}_{X_0}(x_0), \\
    &P^{x_0, \pi_x}[x_{k+1} \in X_{k + 1} \mid x_{k} = x, u_{k} = \pi_x(k, x)] = T(X_{k + 1} \mid x, \pi_x(x, k)).
\end{align*}
With the definition of $P^{x_0, \pi_x}$, we can now make probabilistic statements on the trajectories of System \eqref{eq:System}. 

\begin{definition}[Probabilistic reach-avoid]
    Consider a compact region of interest $X \subset \mathbb{R}^n$.
    Let $R \subset X$ and $O \subset X$ be closed sets with $R \cap O = \emptyset$, $x_0 \in X$ be the initial state, and $H \in \mathbb{N}$ be the time horizon. Then, for a given control strategy $\pi_x$, \emph{probabilistic reach-avoid} is defined as
    \[
        \begin{aligned}
            P_{\mathrm{ra}}(R, O, x_0, \pi_x, H) = P^{x_0, \pi_x}[&\exists k \in \{0,\ldots, H\}, x_k \in R,\\
            &\forall k' \in \{0, \ldots, k\}, x_{k'} \notin O \wedge x_{k'}\in X].
        \end{aligned}
    \]
\end{definition}
Note that the assumption of a given region of interest is standard \cite{van2023syscore, cauchi2019efficiency} and allows one to only focus on the behavior of the system in a bounded set. 
The problem we consider in this work is then formally defined as follows.
\begin{problem}[Verification and controller synthesis]
\label{prob:Synthesis}
Consider a compact region of interest $X \subset \mathbb{R}^n$. Then, for sets $R, O \subset X$, initial state $x_0 \in X$, and time horizon $H$, synthesize a control strategy $\pi^*_x$ such that $\pi^* \in \arg\max_{\pi_x} P_{\mathrm{ra}}(R, O, x_0, \pi_x, H).$
\end{problem}
Problem~\ref{prob:Synthesis} requires one to compute the strategy $\pi_x$ that maximizes the probability that a trajectory of System~\eqref{eq:System} starting from $x_0$ enters $R$ within $H$ time steps and avoids $O$ until $R$ is reached. 
Maximizing the probability is without loss of generality as, as we will show in Section \ref{sec:value_iteration}, the case $\min$ follows similarly.
Note that focusing on probabilistic reach-avoid in Problem \ref{prob:Synthesis} is not limiting. In fact, certifying reach-avoid properties enables verification of more complex temporal logic specifications \cite[Alg. 11]{Baier2008}.

\paragraph*{Approach}
To solve Problem \ref{prob:Synthesis}, in Section \ref{sec:orthogonally_decoupled_interval_markov_decision_processes}, we introduce a class of robust Markov models, called \acrfullpl{odimdp}, where transition probabilities lie in some product ambiguity set. To abstract System \eqref{eq:System} into an \gls{odimdp}, we first partition the region of interest into a grid of discrete regions. Then, in Section \ref{sec:abstraction}, we show that one can take advantage of the product form of an \gls{odimdp} and, for each region, simply store the marginal probabilities of the original system. This allows us to construct the abstraction for System \eqref{eq:System} in a compositional manner, with improved memory complexity and guaranteeing tighter error bounds compared to existing methods. We then show how an optimal strategy for an \gls{odimdp} can be found via linear programming and mapped back to the original system,  guaranteeing efficiency. 
Before formally introducing \glspl{odimdp}, in the next Section, we start by reviewing \glspl{imdp} and the standard approach to abstracting stochastic systems to IMDPs.

\section{Preliminaries on abstractions to Interval Markov Decision Processes}\label{sec:interval_markov_decision_processes}
\Acrfullpl{imdp} (also called bounded-parameter \glspl{mdp}) are a generalization of \glspl{mdp} in which the transition probability distributions between states lie within some independent intervals \cite{givan2000bounded}.
\begin{definition}\label{def:imdp}
An \acrfull{imdp} is a tuple $M = (S, A, \ambiguityset)$ where
\begin{itemize}
    \item $S$ is a finite set of states,
    \item $A$ is a finite set of actions assumed to be available at each state, and
    \item $\ambiguityset=\{\ambiguityset_{s,a}\}_{s\in S, a\in A}$ are sets of feasible transition probability distributions with $\ambiguityset_{s,a} \in \intamb(S)$.
\end{itemize}
\end{definition}
\Glspl{imdp} are commonly used as the abstraction model for complex stochastic systems \cite{lahijanian2015formal, cauchi2019efficiency, jiang2022safe, chatterjee2008model, badings2023robust, skovbekk2023formal}.
In particular, the standard approach to abstract systems of the form of Eqn. \eqref{eq:System} to an \glspl{imdp} is first to assign to $A$ the actions in $U$. Then, the region of interest $X$ is partitioned into regions $\{s_1, \ldots, s_{|S|}\}$ and each region is associated with a state in $S$. In what follows, for the sake of simpler notation, with an abuse, we will denote by $s$ and $t$, that is, for a transition $s \xrightarrow{a} t$, both the abstract states in $S$ and the corresponding regions of $X$. Then, for a state $s\in S$ and action $a \in A$, $\ambiguityset_{s,a}$ is defined as
\begin{align}
\nonumber
\ambiguityset_{s,a}&=\{ \probdist_{s,a} \in \dists(S):  \forall t \in S,  \\
&\min_{x\in s} T(t\mid x,a) \leq \probdist_{s,a}(t) \leq  \max_{x\in s} T(t \mid x,a)\}.
\label{Eqn:AmbiguitySetIMDP}
\end{align}
Once the abstraction is built, then one can rely on existing polynomial time algorithms for \glspl{imdp} to compute reach-avoid probabilities and optimal strategies on the abstraction \gls{imdp}, which can then be mapped back to the concrete system \cite{lahijanian2015formal, cauchi2019efficiency}.
However, the abstraction approach outlined above, while general and sound, has the following drawbacks: 
\begin{itemize}
    \item Memory consumption is often the main bottleneck. In fact, to perform value iteration on an \gls{imdp} $M$, for each pair of regions $s,t \in S$ and $a\in A$ we need to store $\underline{p}_{s,a}(t) = \min_{x\in s} T(t\mid x,a)$ and $\overline{p}_{s,a}(t) = \max_{x\in s} T(t\mid x,a)$, leading to a memory complexity of $O(|S|^2|A|)$.
    \item For any $s,t \in S$, the computation of $\min_{x\in s} T(t\mid x,a)$ and $\max_{x\in s} T(t\mid x,a)$ generally requires approximations, which lead to conservatism. For instance, in the case of Gaussian additive noise $x_{k+1} = f(x_k, u_k) + v_k$, where $v_k \sim \mathcal{N}(0, \Sigma)$ with diagonal covariance matrix $\Sigma \in \mathbb{S}^{n}_{++}$, under the assumption that to each $s\in S$ is associated a hyperrectangular region, i.e., $s=[\underline{s}_1, \overline{s}_1]\times \cdots \times [\underline{s}_n, \overline{s}_n]$, the state-of-the-art approach \cite{cauchi2019efficiency, laurenti2020formal, adams2022formal} is to rely on the following under-approximation for the $\min$ case (similar reasoning holds for the $\max$ case):
    \begin{equation}\label{eq:imdp_product_ambiguity_set}
        \begin{aligned}
        \min_{x\in s} T(t\mid x,a)= 
        &\min_{x\in s} \int_{t} \mathcal{N}(y\mid f(x,a), \Sigma) \,dy \geq  \\
        &\prod_{i=1}^n \min_{x\in s} \int_{\underline{t}_i}^{\overline{t}_{i}} \mathcal{N}(y_i \mid f(x,a)_i, \Sigma_{ii}) \,dy_i,
        \end{aligned}
    \end{equation}
    which introduces conservatism by allowing each marginal to be optimized independently.
    \item The abstraction process does not exploit any structural property of the system to reduce the size of the resulting ambiguity set $\ambiguityset$. For instance, if the noise is independent across the various dimensions (as is the case for System \eqref{eq:System}), one may want to account for this information to reduce the size of the resulting feasible set of distributions in the abstraction.
\end{itemize}
In Section \ref{sec:orthogonally_decoupled_interval_markov_decision_processes}, we will introduce a new subclass of robust \glspl{mdp}, called \gls{odimdp}, where the feasible set of probabilities has a product form and requires only storing the marginal ambiguity sets, i.e. the right-hand side of Eqn. \eqref{eq:imdp_product_ambiguity_set} without the product operator. We will show that by abstracting System \eqref{eq:System} to an \gls{odimdp} one can alleviate the issues identified above and obtain state-of-the-art performance.

\section{Orthogonally Decoupled Interval Markov Decision Processes}\label{sec:orthogonally_decoupled_interval_markov_decision_processes}
Merging ideas from \glspl{imdp} \cite{givan2000bounded} and compositional analysis of \glspl{mdp} \cite{nilsson2018toward}, we propose a new subclass of robust \glspl{mdp}, which we call \emph{\acrfullpl{odimdp}}. Intuitively, in an \glspl{odimdp} for each transition, the ambiguity set of the transition probabilities is a product of marginal interval ambiguity sets.
In what follows, in this section, we first formally introduce \glspl{odimdp}. Then, we show how the memory requirements for this class of models are generally orders of magnitude lower compared to those of \glspl{imdp}. Furthermore, we also prove that they can produce tighter ambiguity sets compared to \glspl{imdp}. These properties will be employed in Section \ref{sec:abstraction} to efficiently abstract System \eqref{eq:System} into an \glspl{odimdp} and solve Problem \ref{prob:Synthesis}.

\begin{definition}\label{def:odimdp}
    An \acrfull{odimdp} with $n$ marginals, also called components, is a tuple $M = (S, A, \ambiguityset)$ where
    \begin{itemize}
        \item $S = S_1 \times \cdots \times S_n$ is a finite set of joint states with $S_i$ being the set of states in marginal $i$,
        \item $A$ is a finite set of actions assumed to be available in each state, and
        \item $\ambiguityset=\{\ambiguityset_{s,a}\}_{s\in S, a\in A}$ are sets of feasible transition probability distributions where $\ambiguityset_{s,a} = \bigotimes_{i=1}^n \ambiguityset^i_{s,a}$ with $\ambiguityset^i_{s,a} \in \intamb(S_i)$.
    \end{itemize}
\end{definition}

Similarly to \glspl{imdp} \cite{lahijanian2015formal}, a \emph{path} of an \gls{odimdp} is a sequence of states and actions $\omega = (s_0,a_0),(s_1,a_1),\dots$, where $(s_k,a_k)\in S \times A$. We denote by $\omega(k) = s_k$ the state of the path at time $k \in \mathbb{N}_0$ and by $\Omega$ the set of all paths.  
A \emph{strategy} or \emph{policy} for an \gls{odimdp} is a function $\pi : S\times \mathbb{N}_0 \to A$ that assigns an action to a given state of an \gls{odimdp}. 
We note that since the ambiguity set is independent for each source-action pair, also called $(s, a)$-rectangularity, any optimal policy is time-varying and deterministic \cite{laurenti2023unifying, suilen2024robust}.
An adversary is a function that assigns a feasible distribution $\probdist \in \ambiguityset$ to a given state \cite{givan2000bounded}. 
Given a strategy and an adversary, an \gls{odimdp} collapses to a finite Markov chain, with the transition probability matrix specified by marginal distributions.

\begin{figure}
    \centering
    \includegraphics[width=0.8\linewidth]{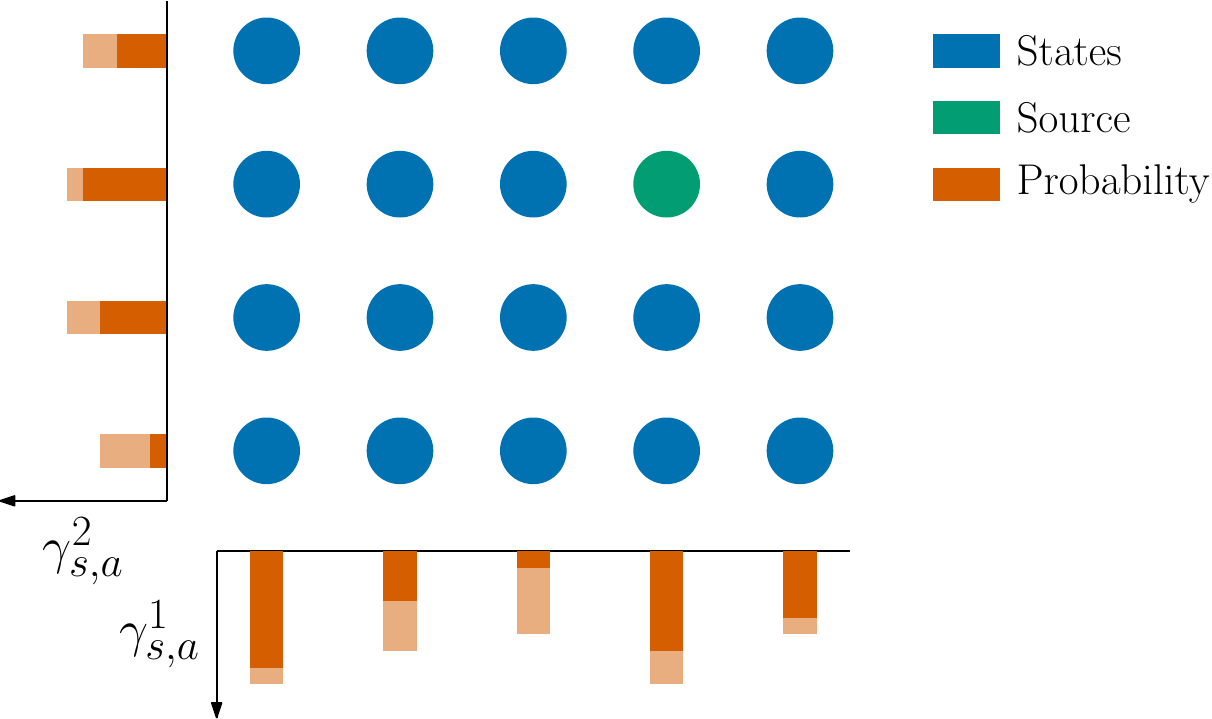}
    \caption{An example of an \gls{odimdp} where, given a source-action pair $(s, a)$ (the green state), the ambiguity set for the transition probability can be decomposed into the product between two independent interval ambiguity sets, $\ambiguityset^1_{s,a}, \ambiguityset^2_{s,a}$. In this example, it is only necessary to store 9 transitions (per source-action pair) in contrast to 20 transitions for a traditional \gls{imdp}.}
    \label{fig:decoupled_imdp_concept}
\end{figure}

\paragraph*{Analysis of space complexity}\label{subsec:space_complexity_analysis}
Before proceeding further, we now discuss the memory requirements to store an \gls{odimdp}. To simplify the presentation, we assume without loss of generality that the number of states in each marginal is equal, that is, $|S_i| = |S_{i+1}|$ for all $i = 1, \ldots, n - 1$.
Hence, we have that the number of states along each marginal is $|S_i| = \sqrt[n]{|S|}$. 
To store an \gls{odimdp}, we need to store the bounds on the transition probability. Then, for each source-action pair, we must store bounds for the ambiguity set along each marginal.
That is, the upper and lower bound for $|S_i|$ states for each marginal. In total, this requires storing $2|S||A|n\sqrt[n]{|S|}$ scalar values.
This is a crucial difference compared to traditional \gls{imdp}, where, as discussed in Section \ref{sec:interval_markov_decision_processes}, for each state-action pair, we need to store $2|S|^2|A|$ scalar values. An example to clarify the different memory requirements between \gls{odimdp} and \gls{imdp} is reported in Fig. \ref{fig:decoupled_imdp_concept}.

\subsection{Comparison of ambiguity sets}\label{subsec:theoretical_results}
At first glance, the ambiguity set of an \glspl{odimdp} may appear equivalent to that of an \gls{imdp} obtained by multiplying the interval bounds of the interval ambiguity sets of each marginal for each source-action pair, as, for instance, is done in Eqn. \eqref{Eqn:AmbiguitySetIMDP} using the approximation step in Eqn. \eqref{eq:imdp_product_ambiguity_set}.
Remarkably, this is, however, a fallacy. 
In fact, multiplying the interval bounds of the marginal ambiguity sets introduces distributions that are not part of the product ambiguity set, as they may not satisfy the additional constraints on the marginals present in \glspl{odimdp}.
To better understand why this theoretical quirk occurs, consider the following example.

\begin{example}\label{exmp:smaller_ambiguity_set}
Fig. \ref{fig:example_margin_joint_ambiguity_sets} shows an \gls{odimdp} represented by an \gls{imdp} for each marginal. The states are numbered to index into the joint probability distribution as $[0, 1]^4$. Furthermore, on the right, an interval ambiguity set is constructed from the \gls{odimdp} by multiplying the marginal bounds together corresponding to the joint transition.
If we consider a distribution $[0.4, 0.3, 0.08, 0.22]^\top$, we see that it is clearly contained in the joint interval ambiguity set on the right of Fig. \ref{fig:example_margin_joint_ambiguity_sets}. 
However, since the value $0.4$ in the first entry is equal to the upper bound for that entry in the ambiguity set, it forces $[0.5, \cdot]$ and $[0.8, \cdot]$ in the distribution for the vertical and horizontal marginals, respectively (because $0.4 = 0.5\cdot 0.8$). For the vertical marginal to sum to one, we must have $[0.5, 0.5]$. The marginal decomposition of the last element $0.22$ thus needs to be $0.5 p = 0.22$, i.e. $p = 0.22/0.5 = 0.44$, which is not contained in the horizontal marginal ambiguity set (upper bound $0.4$). Hence, the distribution contained in the joint interval ambiguity set cannot be factored into two distributions from the marginal ambiguity sets. 
\begin{figure}
    \centering
    \includegraphics[width=0.85\linewidth]{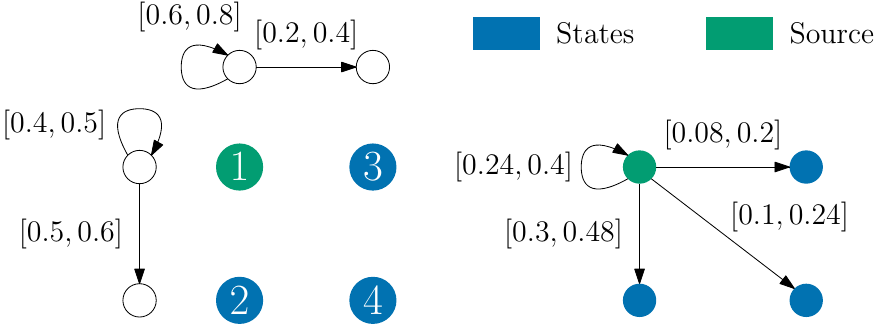}
    \caption{On the left, for the green state, we report two marginal interval ambiguity sets, i.e. an interval ambiguity set for each marginal of a product ambiguity set, with outgoing transitions to all other states. On the right, an \gls{imdp} is constructed by multiplying the interval bounds of the marginal ambiguity sets. By this multiplication of bounds, joint distributions are introduced that cannot be represented as a product of distributions from the marginal ambiguity sets.}
    \label{fig:example_margin_joint_ambiguity_sets}
\end{figure}
\end{example}
The intuition provided in Example \ref{exmp:smaller_ambiguity_set} is formalized in Theorem \ref{thm:ambiguity_set_containment} below, where we show that the ambiguity set of an \gls{odimdp} is contained in that of an \gls{imdp} obtained by multiplying the interval bounds of the interval ambiguity sets of each marginal of the \glspl{odimdp}. This result will allow us to build abstractions for System \eqref{eq:System} that not only require less memory to be stored compared to those described in Section \ref{sec:interval_markov_decision_processes}, but are also guaranteed to provide tighter error bounds for the same partition size of the state space. 

\begin{theorem}\label{thm:ambiguity_set_containment}
For $S = S_1 \times \cdots \times S_n$, consider the interval ambiguity sets $\ambiguityset^1 \in \intamb(S_1),\ldots,$ $\ambiguityset^n \in \intamb(S_n)$, where $\ambiguityset^i=\{ \probdist \in \dists(S_i) : \underline{p}^{i}(s^i) \leq \probdist(s^i) \leq \overline{p}^i(s^i) \text{ for each } s^i \in S_i \}$ . Call $\ambiguityset = \bigotimes_{i=1}^n \ambiguityset^i$. Then, it holds that 
$\ambiguityset \subseteq \bar{\ambiguityset}, $
where 
\[
        \bar{\ambiguityset} = \left\{ \probdist \in \dists(S) : \prod_{i=1}^n \underline{p}^i(s^i) \leq \probdist(s) \leq \prod_{i=1}^n \overline{p}^i(s^i) \text{ for each } s \in S \right\}.
    \]
\end{theorem}
\begin{proof}
    It is sufficient to show that any distribution $\probdist \in \ambiguityset$ is also contained in $\bar{\ambiguityset}$, that is, it satisfies the interval bounds $\prod_{i=1}^n \underline{p}^i(s^i) \leq \probdist(s) \leq \prod_{i=1}^n \overline{p}^i(s^i)$ for each $s \in S$.
    By multiplying the (nonnegative) inequalities $\underline{p}^i(s^i) \leq \probdist^i(s^i) \leq \overline{p}^i(s^i)$ for each component $i = 1, \ldots, n$, it holds for each $s \in S$ that
$            \prod_{i=1}^n \underline{p}^i(s^i) \leq \prod_{i=1}^n \probdist^i(s^i) \leq \prod_{i=1}^n \overline{p}^i(s^i).
$   
     Therefore, we can conclude that the product distribution $\probdist = \bigotimes_{i=1}^n \probdist^i$ is contained in $\bar{\ambiguityset}$. 
\end{proof}

\section{Abstraction}\label{sec:abstraction}
In order to describe the abstraction process of System \ref{eq:System} into an \gls{odimdp}, we first start from the case where the transition kernel is Gaussian and then move to the more general case. In what follows, we assume that the region of interest $X \subset \mathbb{R}^n$ is hyperrectangular\footnote{Note that if the region of interest is not hyperrectangular, but bounded, one can always over-approximate it with a hyperrectangular region.}, that is, $X = [\underline{x}_1, \overline{x}_1] \times \cdots \times [\underline{x}_n, \overline{x}_n]$. 
To abstract a system to an \gls{odimdp}, we start by partitioning $X$ into a grid $\tilde{S}= \tilde{S}_1 \times \cdots \times \tilde{S}_{n_x}$ and associate with each (hyperrectangular) region $s = [\underline{s}_1, \overline{s}_1] \times \cdots \times [\underline{s}_n, \overline{s}_n] \in \tilde{S}$ an abstract state in $\tilde{S}$. Similarly to Section \ref{sec:interval_markov_decision_processes}, we abuse the notation and denote by $s$ both the abstract state in $\tilde{S}$ and the corresponding region in $X$.
The decomposed abstract state is denoted by $s=(s^1, \ldots, s^{n}) \in \tilde{S}$.
In addition, we add one sink state $\breve{s}^i$ for each axis $i$ to represent exiting the set $[\underline{x}_i, \overline{x}_i]$ and denote $S_i = \{\breve{s}^i\} \cup \tilde{S}_i$. 
The reach and avoid states of the abstract model, $R$ and $O$\footnote{We abuse notation slightly by denoting both the concrete reach/avoid regions and the abstract reach/avoid states by $R$ and $O$ respectively.} are respectively the set of states satisfying $s \subseteq R$ and $s \cap O \neq \emptyset$ \cite[Section 4]{cauchi2019efficiency}. All sink states are classified as avoid states. 
Finally, we assign an abstract action $a \in A$ to each $u \in U$. 

In the following subsections, we focus on computing bounds on the transition probability starting with the Gaussian case.

\subsection{Gaussian case}\label{subsec:abstraction_additive_noise_systems}
We start by considering the case of a Gaussian transition kernel:
\begin{equation}
\label{eq:transition_kernel_Gaussian}
    T(X\mid x, u)=  \int_{X} \mathcal{N}(x' \mid \mu(x,u), \Sigma(x, u)) dx'.
\end{equation}

To define the marginal interval ambiguity sets $\ambiguityset^i_{s, a}$, we first need to introduce some notation. In particular, for each $(s,a)$ we define intervals  $[\underline{\mu}_{s, a}^i$, $\overline{\mu}_{s, a}^i]$ and $[\underline{\Sigma}_{s, a}^i,\overline{\Sigma}_{s, a}^i]$ such that for all $x\in s$:
\begin{align}
\label{Eqn:boundMean}
\underline{\mu}_{s, a}^i \leq \mu(x, a)_i \leq \overline{\mu}_{s, a}^i \quad\text{ and }\quad \underline{\Sigma}_{s, a}^i \leq \Sigma(x, a)_{ii} \leq \overline{\Sigma}_{s, a}^i. 
\end{align}
That is, $[\underline{\mu}_{s, a}^i$, $\overline{\mu}_{s, a}^i]$ and $[\underline{\Sigma}_{s, a}^i,\overline{\Sigma}_{s, a}^i]$ represent interval bounds for the mean and variance of the $i$-th marginal of System \eqref{eq:System} starting from region $s$.
Then, computation of $\ambiguityset^i_{s, a}$ reduces to computing for each $t^i \in \tilde{S}_i$ the following transition  probabilities
\begin{align}
    \underline{p}^i_{s,a}(t^i) &= \min_{\mu^i \in [\underline{\mu}_{s, a}^i, \overline{\mu}_{s, a}^i],\; \Sigma^i \in \{\underline{\Sigma}_{s, a}^i, \overline{\Sigma}_{s, a}^i\}} \int_{\underline{t}_i}^{\overline{t}_i} \mathcal{N}(y_i \mid \mu^i, \Sigma^{i}) dy_i, \label{eqn:componentwise_iower_bound}\\
    \overline{p}^i_{s,a}(t^i) &= \max_{\mu^i \in [\underline{\mu}_{s, a}^i, \overline{\mu}_{s, a}^i],\; \Sigma^i \in \{\underline{\Sigma}_{s, a}^i, \overline{\Sigma}_{s, a}^i\}} \int_{\underline{t}_i}^{\overline{t}_i} \mathcal{N}(y_i \mid \mu^i, \Sigma^{i}) dy_i. \label{eqn:componentwise_upper_bound}
\end{align}

\begin{remark}
    Note that analytical solutions for Eqn. \eqref{eqn:componentwise_iower_bound} and \eqref{eqn:componentwise_upper_bound} exist \cite{cauchi2019efficiency, adams2022formal}.
    Furthermore, we stress that the computation of the bounds on mean and variance in Eqn. \eqref{Eqn:boundMean} is a well-studied problem for which there exist tools based on convex optimization \cite{9993724, cauchi2019efficiency, adams2022formal, liu2021algorithms} that can also be applied in the context of Gaussian process regression \cite{patane2022adversarial}.
\end{remark}

For the transition to the sink state, we observe that this is equivalent to the complement of transitioning to inside $X$.
Hence, we may compute the interval bounds via the complement probability.
\begin{align}
    \underline{p}^i_{s,a}(\breve{t}^i) &= 1 - \max_{\substack{\mu^i \in [\underline{\mu}_{s, a}^i, \overline{\mu}_{s, a}^i]\\\Sigma^i \in \{\underline{\Sigma}_{s, a}^i, \overline{\Sigma}_{s, a}^i\}}} \int_{\underline{x}_i}^{\overline{x}_i} \mathcal{N}(y_i \mid \mu^i, \Sigma^{i}) dy_i, \label{eqn:componentwise_iower_bound_outside_state_space}\\
    \overline{p}^i_{s,a}(\breve{t}^i) &= 1 - \min_{\substack{\mu^i \in [\underline{\mu}_{s, a}^i, \overline{\mu}_{s, a}^i]\\\Sigma^i \in \{\underline{\Sigma}_{s, a}^i, \overline{\Sigma}_{s, a}^i\}}} \int_{\underline{x}_i}^{\overline{x}_i} \mathcal{N}(y_i \mid \mu^i, \Sigma^{i}) dy_i. \label{eqn:componentwise_upper_bound_outside_state_space}
\end{align}
Finally, the sink states are absorbing, that is, for all states $s = (s^1, \ldots, s^{i-1}, \breve{s}^i, s^{i+1}, \ldots, s^n)$, we define $\underline{p}^i_{s,a}(\breve{s}^i) = \underline{p}^i_{s,a}(\breve{s}^i) = 1$ and $\underline{p}^i_{s,a}(t^i) = \underline{p}^i_{s,a}(t^i) = 0$ for any $t^i \in \tilde{S}_i$.
From the computation of all interval bounds, we can define $\ambiguityset$ as the product of marginal interval ambiguity sets $\ambiguityset_{s, a}^i$ for each source-action pair $(s, a)$.
Then the \gls{odimdp} is $M = (S, A, \ambiguityset)$.

\subsection{Mixture Models}
Eqn. \eqref{eq:transition_kernel} is a generalization of the Gaussian case with $K > 1$. However, a key additional technical difficulty is that even if the covariance matrix of each Gaussian in the mixture is diagonal, then the covariance matrix for the mixture distribution is not necessarily diagonal.  
Specifically, the joint covariance of the mixture distribution in Eqn. \eqref{eq:transition_kernel}, omitting the dependence on $(x, u)$, is \begin{equation}\label{eq:non-diagonal_joint_covariance}
    \Sigma = \sum_{r = 1}^K \alpha_r \Sigma^r + \sum_{r = 1}^K \alpha_r(\mu^r - \bar{\mu})(\mu^r - \bar{\mu})^\top
\end{equation}
Due to the second term of Eqn. \eqref{eq:non-diagonal_joint_covariance}, the joint covariance $\Sigma$ is generally not diagonal. Consequently, we cannot directly abstract the model to \glspl{odimdp} as it does not have the required product form. To solve this issue, in what follows, we build an \glspl{odimdp} for each Gaussian in the mixture following the approach in Section \ref{subsec:abstraction_additive_noise_systems} and then abstract System \eqref{eq:System} into a mixture of \glspl{odimdp}.

\begin{definition}\label{def:mixture_odimdp}
    A mixture of $K$ \acrshortpl{odimdp} with $n$ marginals is a tuple $M = (S, A, \ambiguityset, \ambiguityset^\alpha)$ where
    \begin{itemize}
        \item $S = S_1 \times \cdots \times S_n$ is a finite set of joint states with $S_i$ being the set of states in marginal $i$,
        \item $A$ is a finite set of actions,
        \item $\ambiguityset=\{\ambiguityset_{r,s,a}\}_{r \in K, s\in S, a\in A}$ are a set  of $K$ product ambiguity sets, where $\ambiguityset_{r,s,a} = \bigotimes_{i=1}^n \ambiguityset^i_{r,s,a}$ with $\ambiguityset^i_{r,s,a} \in \intamb(S_i)$, and
        \item $\ambiguityset^\alpha = \{\ambiguityset^\alpha_{s,a}\}_{s\in S, a\in A}$ are sets of feasible weightings distributions, where $\ambiguityset^\alpha_{s,a} \in \intamb(K)$. A feasible weighting distribution for a source-action pair $(s, a)$ is denoted by $\alpha_{s,a} \in \ambiguityset^\alpha_{s,a}$.
    \end{itemize}
\end{definition}

$M = (S, A, \ambiguityset, \ambiguityset^\alpha)$ can be interpreted as a set of $K$ \glspl{odimdp} sharing the same (decomposed) states $S$ and the same set of actions $A$, but that can have different individual ambiguity sets $\ambiguityset_{r,s,a}$ for each source-action pair $(s, a)$. The weighting $\alpha_{s,a} \in \ambiguityset^\alpha_{s,a}$ defines the probability of selecting a component from the mixture.
Thus, a feasible distribution for $M$ is any distribution $\probdist_{s,a} = \sum_{r \in K} \alpha_{s,a}(r) \probdist^r_{s,a}$, where $\alpha_{s,a} \in \ambiguityset^\alpha_{s,a}$ and $\probdist^r_{s,a} \in \ambiguityset_{r,s,a}$ for each $r \in K$.

The process of abstracting System \eqref{eq:System} to a mixture of \glspl{odimdp}, is as follows:
\begin{enumerate}
    \item $S$ and $A$ are computed similarly to Section \ref{subsec:abstraction_additive_noise_systems} and are shared among all Gaussians in the mixture,
    \item for the $r$-th Gaussian in the mixture follow the approach in Section \ref{subsec:abstraction_additive_noise_systems} to compute $\ambiguityset_{r,s,a}$ for any $s \in S$ and $a\in A$ , and 
    \item for each source-action pair $(s, a)$, the interval ambiguity set $\ambiguityset_{s,a}^\alpha \in \intamb(K)$ is constructed such that for all $x\in s$ there exists  ${\alpha}_{s,a} \in \ambiguityset_{s,a}^\alpha$ with  $
        \alpha^r(x, a) = {\alpha}_{s,a}(r), \text{ for all } r \in K.$
    Similar to the marginal interval ambiguity sets case, we do that by constructing $\ambiguityset_{s,a}^\alpha$ such that for each $r \in K$ for all ${\alpha}_{s,a} \in \ambiguityset_{s,a}^\alpha$ it holds that $
         \min_{x \in s} \alpha^r(x, a) \leq {\alpha}_{s,a}(r) \leq  \max_{x \in s} \alpha^r(x, a).
    $
    
\end{enumerate}

\section{Synthesis of Optimal Policies for odIMDPs}\label{sec:value_iteration}
Once an abstraction of System \eqref{eq:System} into an \glspl{odimdp} or a mixture of \glspl{odimdp} is built, what is left to do is to synthesize an optimal strategy, which can then be mapped back on the original system. In what follows, we consider the \glspl{odimdp} case; the mixture case follows similarly as we will illustrate in Eqn.  \eqref{eq:value_iteration_mixture}.
For an \glspl{odimdp} $M = (S, A, \ambiguityset)$ synthesizing an optimal strategy for $ P_{\mathrm{ra}}(R, O, x_0, \pi, H)$, reduces
to the following optimization problem:
\begin{equation}
\label{eq:goal}
    \begin{aligned}
        \max_{\pi} \min_{\probdist} \mathbb{P}_{\pi, \probdist}[&\omega \in \Omega \mid \; \exists k \in \{0, \ldots, K\}, \, \omega(k)\in R,\\
        &\;\forall k' \in \{0, \ldots, k\}, \omega(k') \notin O \, \wedge \, \omega(k') \in \tilde{S}].
    \end{aligned}
\end{equation}
where $\mathbb{P}_{\pi, \probdist}$ is the probability of the Markov chain induced by strategy $\pi$ and distribution $\probdist$. 
In particular, similar to the \glspl{imdp} case \cite{givan2000bounded, lahijanian2015formal}, Eqn. \eqref{eq:goal} can be computed by solving the following pessimistic robust value iteration:
\begin{equation}\label{Eq:RobustValueIteration}
    \begin{aligned}
        V_{0}(s) &=\mathbf{1}_{R}(s) \\
        V_{k}(s) &=\max_{a \in A} g\left(s, \min_{\probdist_{s,a} \in \ambiguityset_{s,a}} \sum_{t \in S} V_{k-1}(t) \probdist_{s,a}(t) \right)\\
        &=\max_{a \in A}  g\left(s, \min_{\probdist^i_{s,a} \in \ambiguityset^i_{s,a}, \, i=1,\ldots,n} \, \sum_{t \in S} V_{k-1}(t) \prod_{i=1}^{n} \probdist^i_{s,a}(t^i)  \right),
    \end{aligned}
\end{equation}
where $g(s, v) = \mathbf{1}_{R}(s) + \mathbf{1}_{S \setminus T}(s) v$ with $T = R \cup O$ and the last equality is due to the fact that for \glspl{odimdp} $\ambiguityset_{s,a}=\bigotimes_{i=1}^{n} \ambiguityset^i_{s,a} $. That is, the ambiguity set is the product of the marginal ambiguity sets.
If a strategy is given, we denote the value function by $V^\pi_{k}(s)$ and if furthermore the inner minimization is replaced by a maximization, we use the notation $\hat{V}^\pi_k(s)$.
As a consequence of the product structure, solving Eqn. \eqref{Eq:RobustValueIteration} reduces to iteratively solving the following optimization problem:
\begin{equation}\label{Eq:MultilinearRobustValueIteration}
    \min_{\probdist^i_{s,a} \in \ambiguityset^i_{s,a}, \, i=1,\ldots,n} \, \sum_{t \in S} V_{k-1}(t) \prod_{i=1}^{n} \probdist^i_{s,a}(t^i) 
\end{equation}
Because of its multilinear structure, Eqn. \eqref{Eq:MultilinearRobustValueIteration} cannot be solved directly using standard linear programming approaches for \glspl{imdp}, e.g. O-maximization \cite{givan2000bounded}. 
To solve this issue, we develop a divide-and-conquer approach in which we recursively compute a lower bound on Eqn. \eqref{Eq:MultilinearRobustValueIteration} by decomposing the problem into a set of linear programming.
Specifically, we compute an under-approximation of Eqn. \eqref{Eq:MultilinearRobustValueIteration}, 
$ W_{s,a}^{k}$ defined recursively as follows:

\begin{figure}
    \centering
    \includegraphics[width=0.8\linewidth]{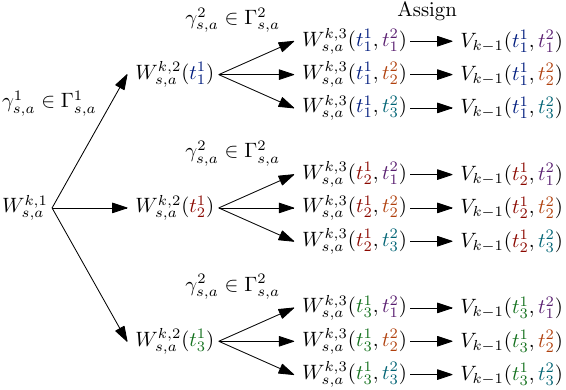}
    \caption{An example of the recursive structure for the proposed algorithm for value iteration over \glspl{odimdp}. The example is for an \gls{odimdp} with two marginals of three states each. Notice going right-to-left that we optimize over $\probdist^2_{s,a}$ three times, as for each subproblem, we assume $t^1$ is given.}
    \label{fig:recursive_structure_value_iteration}
\end{figure}

\begin{equation}\label{Eq:DecoupledRobustValueIteration}
    \begin{aligned}
        W_{s,a}^{k,n}(t^1, \ldots, t^n) &= V_{k - 1}(t)\\
        W_{s,a}^{k,i-1}(t^1, \ldots, t^{i-1}) &= \min_{\probdist^i_{s,a} \in \ambiguityset^i_{s,a}} \sum_{t^i \in S_i} W_{s,a}^{k,i}(
            t^1, \ldots, t^i) \probdist^i_{s,a}(t^i),\\
            &\qquad \qquad \text{ for } i = 2, \ldots, n \\
      W_{s,a}^{k} :=  W_{s,a}^{k,0} &= \min_{\probdist^1_{s,a} \in \ambiguityset^1_{s,a}} \sum_{t^1 \in S_1} W_{s,a}^{k,1}(t^1) \probdist^1_{s,a}(t^1).
    \end{aligned}
\end{equation}
Fig. \ref{fig:recursive_structure_value_iteration} gives an example of Eqn. \eqref{Eq:DecoupledRobustValueIteration} for a 2-marginals \gls{odimdp} with three states in each marginal. The intuition is that we derive a bound for Eqn. \eqref{Eq:MultilinearRobustValueIteration} by solving a set of simpler optimization problems, one marginal at a time.
In particular, $W_{s,a}^{k,i-1}(t^1, \ldots, t^{i-1})$ optimizes the expectation of $W_{s,a}^{k,i}$, which is itself a bound of $V_{k - 1}$, for the marginal $i$, while keeping the values of the other marginals fixed to $t^1, \ldots, t^{i-1}$. 
Note that, as illustrated in Fig. \ref{fig:recursive_structure_value_iteration}, the resulting algorithm is exponential in $n$. However, as each of the optimization problems is a particularly simple linear problem that can be solved efficiently using, e.g., the O-maximization algorithm \cite{lahijanian2015formal, givan2000bounded} and as we will show in Section \ref{subsec:time_complexity_analysis}, the resulting computational time complexity will still be lower compared to the approach described in Section \ref{sec:interval_markov_decision_processes} based on IMDPs. First, in the rest of this Section, in Theorem \ref{thm:divide_and_conquer_algorithm_sound} and Proposition \ref{prop:divide_and_conquer_algorithm_distribution_containment} we show, respectively, that the approach in Eqn. \eqref{Eq:DecoupledRobustValueIteration} is sound and guaranteed to improve the tightness of the results compared to the approach described in Section \ref{sec:interval_markov_decision_processes}. 

\begin{theorem}\label{thm:divide_and_conquer_algorithm_sound}
    For any source-action pair $(s, a) \in S \times A$, it holds that \[
        W_{s, a}^{k,0} \leq \min_{\probdist^i_{s, a} \in \ambiguityset^i_{s, a}, i=1,\ldots,n} \sum_{t \in S} V_{k-1}(t) \prod_{i=1}^{n} \probdist^i_{s, a}(t^i).
    \]
\end{theorem}
\begin{proof}
\begin{align*}
& W_{s, a}^{k,0} = \min_{\probdist^1_{s,a} \in \ambiguityset^1_{s,a}} \sum_{t^1 \in S_1} W_{s,a}^{k,1}(t^1) \probdist^1_{s,a}(t^1) = \\ 
& \min_{\probdist^1_{s,a} \in \ambiguityset^1_{s,a}} \sum_{t^1 \in S_1} \min_{\probdist^2_{s,a} \in \ambiguityset^2_{s,a}} \sum_{t^2 \in S_2}  W_{s,a}^{k,2}(t^1,t^2) \probdist^1_{s,a}(t^1)\probdist^2_{s,a}(t^2) \leq \\
& \min_{\probdist^1_{s,a} \in \ambiguityset^1_{s,a}} \min_{\probdist^2_{s,a} \in \ambiguityset^2_{s,a}} \sum_{t^1 \in S_1} \sum_{t^2 \in S_2}  W_{s,a}^{k,2}(t^1,t^2) \probdist^1_{s,a}(t^1)\probdist^2_{s,a}(t^2) \leq 
\\
&\min_{\probdist^i_{s, a} \in \ambiguityset^i_{s, a}, i=1,\ldots,n} \sum_{t^1 \in S_1}... \sum_{t^n \in S_n}  W_{s,a}^{k,n}(t^1,...,t^n) \prod_{i=1}^{n} \probdist^i_{s, a}(t^i) =\\
&\min_{\probdist^i_{s, a} \in \ambiguityset^i_{s, a}, i=1,\ldots,n} \sum_{t \in S} V_{k-1}(t) \prod_{i=1}^{n} \probdist^i_{s, a}(t^i).
\end{align*}
\end{proof}

Note that an alternative option to the divide-and-conquer algorithm considered here would be to simply multiply the interval bounds of the marginal ambiguity sets for each source-action pair, obtaining an ambiguity set identical to the one in Eqn. \eqref{Eqn:AmbiguitySetIMDP}. As already mentioned, we do not do that because it would increase the conservativism of the resulting approach. This is proved in the following proposition that guarantees that our approach in Eqn. \eqref{Eq:DecoupledRobustValueIteration} is guaranteed not to be more conservative than the one described in Section \ref{sec:interval_markov_decision_processes}. In practice, in the experimental section, we show how our approach consistently returns substantially tighter bounds. 

\begin{proposition}\label{prop:divide_and_conquer_algorithm_distribution_containment}
For an \glspl{odimdp} $M = (S, A, \ambiguityset)$ consider the function $V_{k-1} : S \to \mathbb{R}$ and define the interval ambiguity set associated to $M$ as $    \bar{\ambiguityset}_{s, a} = \left\{ \probdist \in \dists(S) : \underline{p}_{s,a}(t) \leq \probdist_{s, a}(t) \leq \overline{p}_{s,a}(t) \text{ for each } t \in S \right\}, $
    where $\underline{p}_{s,a}(t) = \prod_{i=1}^n \underline{p}_{s, a}^i(t^i)$ and $\overline{p}_{s,a}(t) = \prod_{i=1}^n \overline{p}_{s, a}^i(t^i)$
    Then, for any $(s, a) \in S \times A$ it holds that $W^k_{s,a}\geq \min_{\probdist_{s,a} \in \bar{\ambiguityset}_{s,a} } \, \sum_{t \in S} V_{k-1}(t)  \probdist_{s,a}(t^i).$
\end{proposition}
\begin{proof}
To conclude the proof, it suffices to show that for any source-action pair $(s, a) \in S \times A$, the joint distribution $\probdist_{s, a}(t)$ obtained by optimizing Eqn. \eqref{Eq:DecoupledRobustValueIteration} is contained $\bar{\ambiguityset}_{s, a}$. By definition of sets $\ambiguityset^i_{s,a}$ in Eqn. \eqref{Eq:DecoupledRobustValueIteration} it holds that the $i$-th marginal of $\probdist_{s, a}(t)$ must be contained within $[\underline{p}_{s, a}^i(t^i), \overline{p}_{s, a}^i(t^i)]$. This implies that $
        \prod_{i=1}^n \underline{p}_{s, a}^i(t^i) \leq \probdist_{s, a}(t) \leq \prod_{i=1}^n \overline{p}_{s, a}^i(t^i). $
\end{proof}

\begin{table*}
    \centering
    \caption{A summary of the benchmarks and their characteristics. Note that the number of states for the different experiments is the specified amount unless otherwise stated.}
    {\small\begin{tabular}{lclcccc}
        \toprule
        Name & Dimension & Dynamics type & Property & \# inputs & \# states \\
        \midrule
        Car parking \cite{van2023syscore} & 2 & Additive linear & Reach-avoid & 9 & 1600 \\
        Robot reachability \cite{10.1007/978-3-031-68416-6_15} & 2 & Additive linear (non-linear control) & Reachability & 121 & 400 \\
        Robot reach-avoid  \cite{10.1007/978-3-031-68416-6_15} & 2 & Additive linear (non-linear control) & Reach-avoid & 441 & 1600 \\
        Building automation system  \cite{10.1007/978-3-031-68416-6_15} & 4 & Additive affine & Safety & 4 & 1225 \\
        Van der Pol \cite{van2023syscore} & 2 & Additive polynomial & Reachability & 11 & 2500 \\
        NNDM Cartpole & 4 & Additive \gls{nndm} & Safety & 2 & 192000 \\
        6D linear model & 6 & Additive linear & Safety & 0 & 262144 \\
        7D linear model & 7 & Additive linear & Safety & 0 & 2097152 \\
        Dubin's car GP \cite{reed2023promises} & 3 & Gaussian process with deep kernel learning & Reach-avoid & 7 & 25600 \\
        Stochastic switched linear & 2 & Gaussian mixture, each linear & Reach-avoid & 0 & 1600 \\
        \bottomrule
    \end{tabular}}
    \label{tab:benchmark_description}
\end{table*}

The proposed divide-and-conquer algorithm extends readily to mixtures of \glspl{odimdp} by executing the algorithm in Eqn. \eqref{Eq:DecoupledRobustValueIteration} for each element of the mixture and optimizing for the worst mixture distribution. That is, by solving:
\begin{equation}\label{eq:value_iteration_mixture}
    W_{s,a}^{k} := \min_{\alpha \in \ambiguityset^\alpha_{s,a}} \sum_{r\in K}\alpha(r) W_{r,s,a}^{k,0},
\end{equation}
where $ W_{r,s,a}^{k,0}$ is computed as in Eqn. \eqref{Eq:DecoupledRobustValueIteration}. 
As $\ambiguityset^\alpha_{s,a}$ is an interval ambiguity set, the above is a linear program that can be efficiently solved using O-maximization \cite{givan2000bounded}.

\paragraph*{Analysis of time complexity}\label{subsec:time_complexity_analysis}
Similarly to the analysis of the space complexity (see Subsection \ref{subsec:space_complexity_analysis}), we assume that the number of states in each marginal is equal, i.e. $|S_i| = \sqrt[m]{|S|}$.
Since value iteration is performed independently for each source-action pair, to compute the time complexity of the algorithm, we can analyze each pair separately.
To do this, we start by observing the number of (branching) internal nodes in the recursion tree required to solve Eqn. \eqref{Eq:DecoupledRobustValueIteration} (see Fig. \ref{fig:recursive_structure_value_iteration}) is a geometric series with base $|S_i|$ up to exponent $n - 1$. %
The complexity of the geometric series is $O\left(\left(|S_i|\right)^{n-1}\right)$ \cite{Cormen2001introduction}.
At each internal node of the tree, the O-maximization algorithm \cite{lahijanian2015formal} is executed once on a set of size $|S_i|$.
O-maximization is limited by the required sorting, which is of complexity $O(|S_i| \lg |S_i|)$.
Now, since $|S_i| = \sqrt[n]{|S|} = |S|^\frac{1}{n}$, the complexity for each source-action pair is $O(|S|\lg |S_i|)$, and the full complexity is thus $O(|S|^2|A| \lg\sqrt[n]{|S|})$.
In contrast, value iteration for a regular \gls{imdp} \cite{givan2000bounded} has a time complexity of $O(|S|^2 |A| \lg |S|)$.

\subsection{Correctness}\label{subsec:correctness}
In this subsection, we map the strategy $\pi$ to a switching strategy $\pi_x : \mathbb{R}^n \times \mathbb{N}_0 \to U$ for System \eqref{eq:System}, ensuring the correctness of our aproach. To this end, define the function $J : \mathbb{R}^n \to S$ that maps each concrete state to an abstract state and its corresponding region. That is, $J(x) = s$ if and only if $x \in s$. Then the switching strategy for System \eqref{eq:System} is defined as $\pi_x(x, k) = \pi(J(x), k)$ for all $x \in \mathbb{R}^n$. The following theorem ensures the correctness of $\pi_x$.

\begin{theorem}
    Let $M$ be an \gls{odimdp} abstraction of System \eqref{eq:System}. Then, for any strategy $\pi$, it holds that for any $x_o\in X$ \[
        P_{\mathrm{ra}}(R, O, x_0, \pi_x, H) \in [V^\pi_H(J(x_0)), \hat{V}^\pi_H(J(x_0))],
    \]
    where $V^\pi_H(s), \hat{V}^\pi_H(s)$ are the lower and upper bounds on the satisfaction probability for $M$, as introduced in Section \ref{sec:value_iteration}.
\end{theorem}
The proof follows similarly to \cite[Theorem 4]{laurenti2023unifying} by induction on Eqn. \eqref{Eq:RobustValueIteration} and relying on Theorem \ref{thm:divide_and_conquer_algorithm_sound} to guarantee the soundness of our approach at any time step.
Note that if the regions $R$ and $O$ do not align with the partitioning, then extra care is needed for computing the upper bound on satisfaction \cite[Section 7]{cauchi2019efficiency}.

\section{Experiments}\label{sec:experiments}
To show the efficacy of our approach, we conducted a wide range of empirical studies with benchmarks ranging from linear 2D systems to nonlinear systems, linear 7D systems, \acrfullpl{nndm}, and Gaussian mixture models. A summary of the benchmarks can be found in Table \ref{tab:benchmark_description} and extended descriptions are provided in Appendix \ref{app:experiment_details}.
All benchmarks are verified for a time horizon $H = 10$ unless otherwise specified. 
We compare our results against state-of-the-art tools for abstractions to \glspl{imdp} and \glspl{mdp}. In particular, for \glspl{imdp}, we compare with the method in \cite[Theorem 1]{adams2022formal}, as well as IMPaCT \cite{10.1007/978-3-031-68416-6_15}, which uses nonlinear optimization and Monte Carlo integration to compute the transition probabilities needed to compute the abstraction as illustrated in Section \ref{sec:interval_markov_decision_processes}. 
Furthermore, we also compare with SySCoRe \cite{van2023syscore}, which is based on model reduction and abstractions to \glspl{mdp}.
To run value iteration for both the proposed method and the method of \cite{adams2022formal}, we use \texttt{IntervalMDP.jl} \cite{mathiesen2024intervalmdpjlacceleratedvalueiteration}, which is our package for parallelized value iteration over \glspl{imdp} and \glspl{odimdp} written in Julia \cite{bezanson2012julia}.
The experiments are all run on the TU Delft supercomputer, DelftBlue \cite{DHPC2024}, on one memory partition node with an Intel Xeon E5-6248R CPU (12 cores allocated) and 200 GB memory.
We stress that despite recent advances in performing value iteration on GPUs \cite{mathiesen2024intervalmdpjlacceleratedvalueiteration, 10.1007/978-3-031-68416-6_15}, the experiments are all conducted on CPUs.

\subsection{Comparison with IMDP-based approaches}\label{sec:experiments_imdp_comparison}
To compare abstraction approaches based on \glspl{odimdp} and \glspl{imdp}, for each method, %
 we report abstraction and certification time and memory usage. 
Furthermore, to quantify the conservatism of the results, we consider the following metrics: (i) the mean lower bound $\mathrm{mean}(\{V^\pi_H(s): s \in S \setminus T\})$ and (ii) the mean error, $\varepsilon = \mathrm{mean}(\{{\hat{V}}^\pi_H(s) - V^\pi_H(s): s \in S \setminus T\})$ where $T$ is the set of terminal states and $\pi$ is the optimal strategy for Eqn. \eqref{Eq:DecoupledRobustValueIteration}. ${\hat{V}}^\pi_H$ and  $V^\pi_H$ are the lower and upper bounds on the satisfaction probability returned by our approach. That is, the solution of Eqn. \eqref{eq:goal} using the approach in Eqn. \eqref{Eq:DecoupledRobustValueIteration}.
Furthermore, we report the (aggregated) difference in $V$ between the methods as $\mathrm{agg}(\{V_H^{\mathrm{odIMDP}}(s) - V_H^{\mathrm{other}}(s) : s\in S \setminus T\})$ where the aggregation functions are $\mathrm{agg} \in \{\min, \max, \mathrm{mean}\}$ and the methods are $\mathrm{other} \in \{\mathrm{IMDP}, \mathrm{IMPaCT}\}$ ($\mathrm{IMDP}$ refers to the method in \cite[Theorem 1]{adams2022formal}).
We consider the same discretization size in the abstraction process for all methods, reported in Table \ref{tab:benchmark_description}.

The results of the comparison with \gls{imdp}-based approaches can be found in Table \ref{tab:performance_results} and Table \ref{tab:satisfaction_results} for computation performance and satisfaction probability, respectively.
Comparing the computation time of \cite[Theorem 1]{adams2022formal} and IMPaCT \cite{10.1007/978-3-031-68416-6_15} with that of our method based on \glspl{odimdp}, the abstraction time of \glspl{odimdp} is at least 5x faster, up to 80x faster, for 2D benchmarks, and at least two orders of magnitude faster for >3D benchmarks.
The certification time is comparable to or faster than both \cite[Theorem 1]{adams2022formal} and IMPaCT.
As expected, the memory requirements are considerably lower for \glspl{odimdp}: Compared to IMPaCT, \glspl{odimdp} use at least an order of magnitude less memory. Most remarkably, for the 6D and 7D linear models, \glspl{odimdp} uses 4682x and 30476x less memory, respectively, compared to IMPaCT. If we compare the satisfaction probability results in Table \ref{tab:satisfaction_results}, the proposed method is at least as good as the method \cite[Theorem 1]{adams2022formal} for all states in all examples. %
This ranges from almost the same satisfaction probability for the 2D robot with a reach-avoid specification to $17.33$ percentage points better on average for the 4D building automation system.
A similar analysis holds for IMPaCT, with the caveat that IMPACT uses non-linear optimization to compute Eqn. \eqref{Eqn:AmbiguitySetIMDP}, thus its results are generally tighter compared to those of \cite[Theorem 1]{adams2022formal}.

\begin{table*}
    \centering
    \caption{Computation time and memory requirements for \gls{imdp}-based approaches across different benchmarks. Time is measured in seconds, and memory is measured in MB. OOM denotes out of memory. To compare the memory requirements for the benchmarks with OOM, we have calculated the amount of memory required to store the transition interval bounds (as a dense matrix) using the formula for \glspl{imdp} in Subsection \ref{subsec:space_complexity_analysis}.}
    {\small\begin{tabular}{l|rrr|rrr|rrr}
         \toprule
                   & \multicolumn{3}{l|}{Our method} & \multicolumn{3}{l|}{\citeauthor{adams2022formal} \cite{adams2022formal}} & \multicolumn{3}{l}{IMPaCT \cite{10.1007/978-3-031-68416-6_15}} \\
         Benchmark & Abs. time & Cert. time & Mem. & Abs. time & Cert. time & Mem. & Abs. time & Cert. time & Mem. \\\midrule
         Car parking & 0.150 & 0.146	& 19.2 & 13.497 & 0.255 & 138.1 & 19.570 & 0.846 & 304.9 \\
         Robot reachability & 0.665 & 0.168 & 21.6 & 12.659 & 0.129 & 88.0 & 20.611 & 0.765 & 306.7 \\
         Robot reach-avoid & 14.259 & 7.641 & 547.2 & 1136.720 & 6.739 & 4143.8 & 918.856 & 37.526 & 16388.0 \\
         Building automation system & 0.023 & 0.237 & 3.1 & 7.273 & 0.285 & 105.1 & 17.564 & 0.318 & 96.0 \\
         Van der Pol & 1.658 & 0.257 & 45.5 & 27.255 & 0.827 & 353.6 & 113.235 & 3.233 & 1093.4 \\
         NNDM Cartpole & 236.154 & 550.747 & 610.8 & 42326.400 & 3.751 & 1590.9 & \multicolumn{3}{c}{Incompatible dynamics}\\
         6D linear model & 8.959 & 359.887 & 237.3 & \multicolumn{3}{c|}{OOM} & \multicolumn{3}{c}{OOM ($\approx 1.1TB$)}\\
         7D linear model & 66.231 & 13903.540 & 2310.8 & \multicolumn{3}{c|}{Timeout (24h)} & \multicolumn{3}{c}{OOM ($\approx 70.4TB$)}\\
         Dubin's car GP & 13.816 & 19.562 & 352.2 & 336.940 & 31.333 & 14265.8 & \multicolumn{3}{c}{Incompatible dynamics}\\
         Stochastic switched linear & 0.033 & 0.045 & 4.5 & 3.391 & 0.038 & 41.0 & \multicolumn{3}{c}{NLopt failure} \\
         \bottomrule
    \end{tabular}}
    \label{tab:performance_results}
\end{table*}

\begin{table*}
    \centering
    \caption{Satisfaction probabilities for \gls{imdp}-based approaches. $V$ denotes the lower bound satisfaction probability and $\epsilon$ the mean error. For \citeauthor{adams2022formal} \cite{adams2022formal} and IMPaCT \cite{10.1007/978-3-031-68416-6_15}, we also report the (minimum, maximum, and mean over regions) difference $\delta$ in $V$ to our method, where positive means that our method yields a higher satisfaction probability and vice versa. The unreliable result for IMPaCT with the Van der Pol benchmark is due to it returning \texttt{NaN} for every state.}
    {\small\begin{tabular}{l|rr|rrrrr|rrrrr}
         \toprule
         & \multicolumn{2}{l|}{Our method} & \multicolumn{5}{l|}{\citeauthor{adams2022formal} \cite{adams2022formal}} & \multicolumn{5}{l}{IMPaCT \cite{10.1007/978-3-031-68416-6_15}} \\
         Benchmark & Mean $V$ & $\varepsilon$ & Mean $V$ & $\varepsilon$ & Min $\delta$ & Max $\delta$ & Mean $\delta$ & Mean $V$ & $\varepsilon$ & Min $\delta$ & Max $\delta$ & Mean $\delta$ \\\midrule
         Car parking & 0.269 & 0.3885 & 0.213 & 0.5315 & 0.0040 & 0.1428 & 0.0560 & 0.213 & 0.5183 & 0.0040 & 0.1422 & 0.0556 \\
         Robot reachability & 0.889 & 0.1108 & 0.881 & 0.1186 & 0.0060 & 0.0119 & 0.0079 & 0.890 & 0.1098 & -0.0058 & 0.0022 & -0.0010 \\
         Robot reach-avoid & 0.980 & 0.0199 & 0.979 & 0.0208 & 0.0005 & 0.0027 & 0.0008 & 0.980 & 0.0202 & -0.0001 & 0.0018 & 0.0003 \\
         Building automation system & 0.263 & 0.7336 & 0.090 & 0.9076 & 0.0510 & 0.2297 & 0.1733 & 0.174 & 0.8237 & 0.0304 & 0.1131 & 0.0897 \\
         Van der Pol & 0.069 & 0.3367 & 0.051 & 0.4178 & 0.0000 & 0.0529 & 0.0177 & \multicolumn{5}{c}{Unreliable results} \\
         NNDM Cartpole & 0.004 & 0.7634 & 0.000 & 0.7184 & 0.0000 & 0.4101 & 0.0037 & \multicolumn{5}{c}{Incompatible dynamics}\\
         6D linear model & 0.958 & 0.0419 & \multicolumn{5}{c|}{OOM} & \multicolumn{5}{c}{OOM}\\
         7D linear model & 0.952 & 0.0483 & \multicolumn{5}{c|}{Timeout} & \multicolumn{5}{c}{OOM}\\
         Dubin's car GP & 0.362 & 0.3461 & 0.216 & 0.5046 & 0.0000 & 0.8383 & 0.1458 & \multicolumn{5}{c}{Incompatible dynamics}\\
         Stochastic switched linear & 0.411 & 0.2828 & 0.366& 0.3605 & 0.0000 & 0.0979 & 0.0456 & \multicolumn{5}{c}{NLopt failure} \\
         \bottomrule
    \end{tabular}}
    \label{tab:satisfaction_results}
\end{table*}

\subsection{Comparison with SySCoRe (MDP-based approach)}
We now also compare with SySCoRe\footnote{We only compare with the finite-state abstraction of SySCoRe, i.e. without model order reductions, to assess the impact of the model choice. As mentioned in Section \ref{sec:introdution}, model order reductions are also interesting in the context of \gls{imdp}-based abstractions and is subject to future work.} \cite{van2023syscore}, which represents a state-of-the-art tool for control and verification of stochastic systems based on abstractions to \glspl{mdp}.
Note that, as SySCoRe only supports gridding of 2D systems, we cannot compare on systems of higher dimension. In particular, in Fig. \ref{fig:comparison_with_mdp}, we consider the car parking benchmark introduced in Table \ref{tab:benchmark_description}, and for various sizes of the abstraction, we report the error $\varepsilon$, as introduced in the previous subsection. From the figure, we have three key observations:
\begin{itemize}
    \item SySCoRe requires a relatively large number of regions in the partition before its mean error decreases with larger time horizons.
    \item \glspl{odimdp} tends to obtain substantially tighter bounds for the same abstration size. For example, the mean error at time $100$ of SySCoRe with an \gls{mdp} of size $90k$ is larger than that of our approach with an \glspl{odimdp} of $2.5k$ states.
    \item The mean error for \glspl{odimdp} first increases with time horizons to a level slightly above SySCoRe with the same number of regions and then decreases for longer horizons.
\end{itemize}
We should, however, stress that SySCoRe requires significantly less memory compared to our approach for the same partition size. The memory usage compared to the mean error at convergence for both SySCoRe and \glspl{odimdp} is reported in Table \ref{tab:syscore_mem_comparison}. Nevertheless, especially for longer time horizon, because of the fewer regions per dimension required by our approach to achieve a similar level of error, we expect our approach based on \glspl{odimdp} to be substantially more scalable for higher dimensions. In fact, this has already been observed in \cite{cauchi2019efficiency} for \gls{imdp}-based approaches compared to \gls{mdp}-based approaches and, as illustrated in Table \ref{tab:performance_results}, our approach substantially outperforms that of \cite{cauchi2019efficiency} in terms of scalability. 

\begin{figure}
    \centering
    \includegraphics[width=\linewidth]{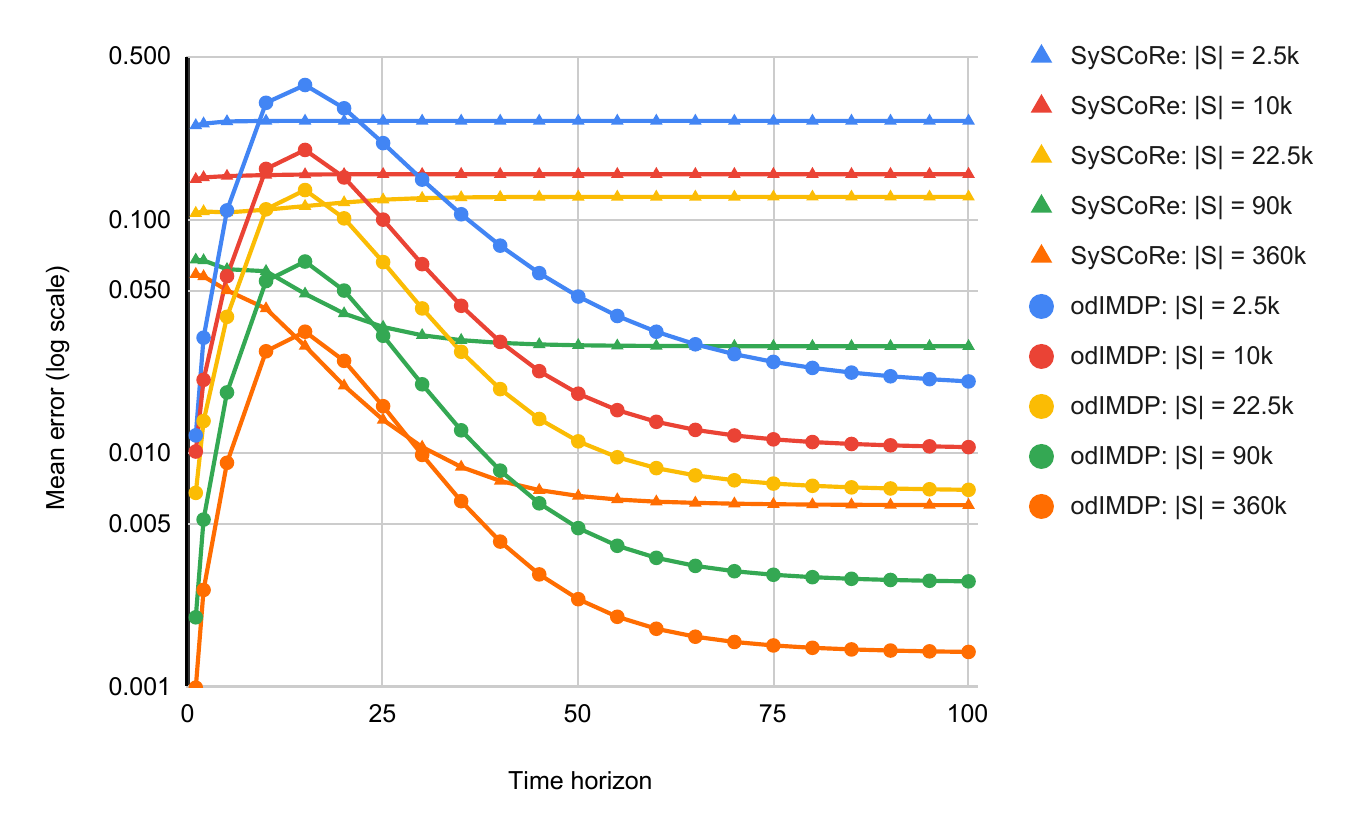}
    \caption{Comparing $\varepsilon$ for a varying number of regions on the abstraction (assuming a uniform partition of the state space) and time horizons for both SySCoRe \cite{van2023syscore} and \glspl{odimdp}.}
    \label{fig:comparison_with_mdp}
\end{figure}

\begin{table}
    \centering
    \caption{The memory usage and the mean error $\varepsilon$ at 100 time steps for SySCoRe \cite{van2023syscore} and \glspl{odimdp} on the car parking benchmark (see Table \ref{tab:benchmark_description}).}
    \label{tab:syscore_mem_comparison}
    \begin{tabular}{l|cc|cc}
        \toprule
        \# regions & \multicolumn{2}{l|}{SySCoRe} & \multicolumn{2}{l}{\glspl{odimdp}} \\
        \cmidrule{2-5}
        & Mem. (MB) & $\varepsilon$ & Mem. (MB) & $\varepsilon$ \\
        \midrule
        2500 & 0.80 & 0.2650 & 35.55 & 0.0204 \\
        10000 & 3.20 & 0.1568 & 279.51 & 0.0107 \\
        22500 & 7.20 & 0.1254 & 937.86 & 0.0070 \\
        90000 & 28.81 & 0.0288 & 7459.02 & 0.0028 \\
        360000 & 115.21 & 0.0060 & 59499.08 & 0.0014 \\
        \bottomrule
    \end{tabular}
\end{table}

\subsection{Empirical convergence analysis}
As hinted in Fig. \ref{fig:comparison_with_mdp}, the error of our approach is reduced with finer partitioning of the region of interest $X$. 
To analyze this behavior in more detail, in Fig. \ref{fig:convergence_plot} we report the mean error and 95\% confidence interval as a function of the number of regions per axis for the car parking and building automation system benchmarks.
In both cases, the mean error decreases with more regions for both systems. The slight increase at various points for the car parking benchmark is attributed to a misalignment of the reach and avoid regions with the partitioning of $X$ (see Section \ref{subsec:correctness} for more details). 
The confidence interval for the car parking benchmark suggests that while a larger number of regions quickly reach low errors, some hardly converge. Future work will study the convergence from a theoretical perspective.

\begin{figure}
    \centering
    \includegraphics[width=\linewidth]{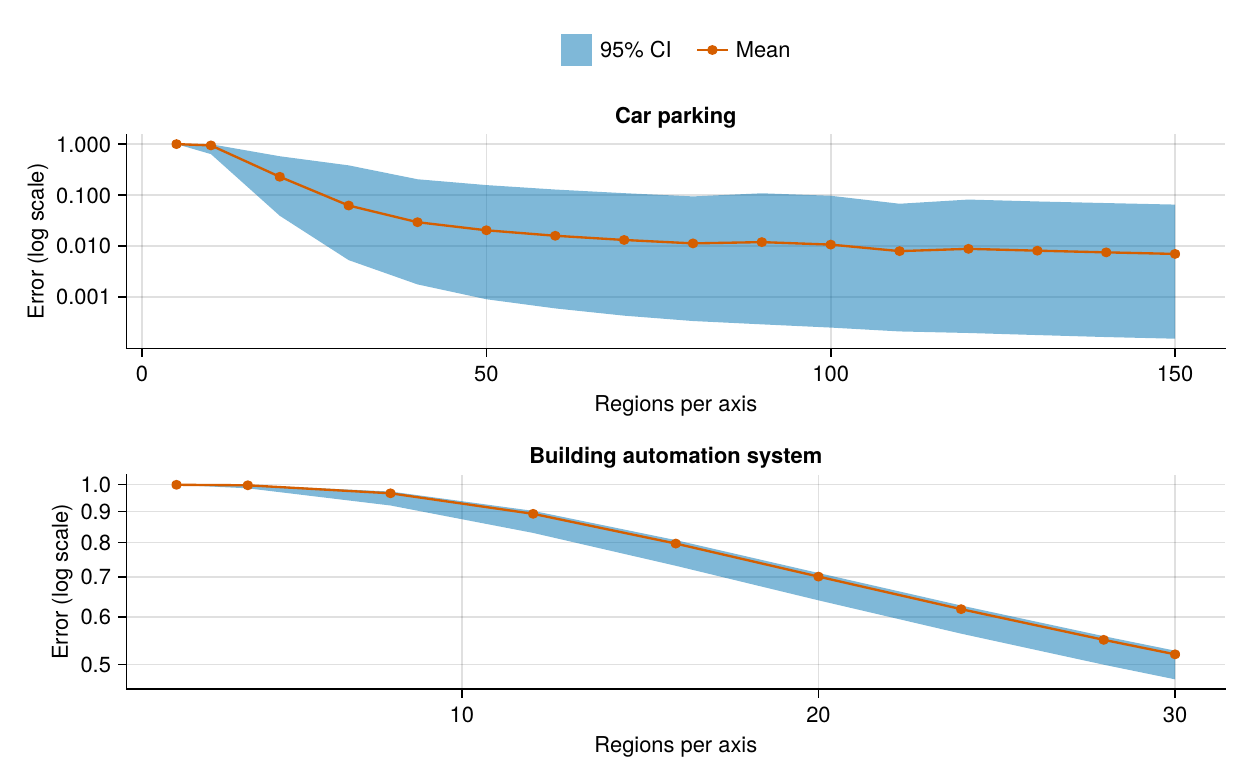}
    \caption{Mean error and 95$\%$ confidence interval, with respect to a uniform distribution of initial conditions, for various sizes of partitions for the car parking and building automation system benchmarks.}
    \label{fig:convergence_plot}
\end{figure}

\section{Conclusion}\label{sec:conclusion}
We presented a novel approach for control synthesis of non-linear stochastic systems against probabilistic reach-avoid specifications. Our approach is based on abstractions to \acrfull{odimdp}, which is a new class of robust Markov models with uncertain transition probabilities having the product form. We showed how such a structure on the transition probabilities allows one to abstract a large class of stochastic systems by obtaining substantial improvements both in terms of memory requirements and tightness of the results compared to state-of-the-art. The theoretical findings are supported by experimental results, showing how our new approach can successfully outperform existing competing methods in various tasks.  %

Naturally, our proposed approach is not without limitations.
First, our framework only applies to systems with independence in the stochasticity across dimensions. As discussed, this includes various systems commonly used in practice. However, if there is dependence across the marginals, then the transition probabilities do not have the product form. Consequently, our framework cannot be applied. How to extend our approach to this setting is an important future direction.
Secondly, exactly and efficiently solving the multi-linear problem of Eqn. \eqref{Eq:MultilinearRobustValueIteration} is still an open problem for which we employed relaxations.
Finally, our approach assumes on a grid-partitioning of the region of interest, although the literature \cite{adams2022formal, reed2023promises} on refinement has shown that heterogeneous abstractions can be beneficial. 

\bibliographystyle{ACM-Reference-Format}
\bibliography{bib.bib}             %
                      
\clearpage
\appendix

\section{Experiment details}\label{app:experiment_details}

\subsection{Other tools for verification of stochastic systems}
We have explored comparing against SReachTools \cite{SReachTools}, however, this tool supports a different class of problems.
Namely, that \emph{given an initial region}, find a reachability tube with high probability, or given a target tube, find a set of initial states with a minimum probability of satisfaction.
We have also tested StocHy \cite{cauchi2019stochyautomatedverificationsynthesis}, but the tool crashed in 3 of its 4 case studies, suggesting that the tool is not maintained and thus we decided against comparing with StocHy.

\subsection{Benchmarks}
\paragraph*{Car parking}
A 2D case study from \cite{van2023syscore} of parking a car with dynamics $x_{k + 1} = 0.9x_k + 0.7u_k + v_k$ with $v_k \sim \mathcal{N}(0, I)$. The region of interest is $X = [-10, 10]^2$ and the input space $U = \{-1, 0, 1\}^2$. 
The reach region is $R = [4, 10] \times [-4, 0]$ and the avoid region is $O = [4, 10] \times [0, 4]$. The goal is to maximize the pessimistic reach-avoid probability over 10 time steps.
We consider a partition of the  region of interest of $|\tilde{S}_1| \times |\tilde{S}_2| = (40, 40)$.

\paragraph*{2D robot}
A 2D case study from \cite{10.1007/978-3-031-68416-6_15} of a planar robot with nonlinear input \[
    x_{k + 1} = x_k + 10u_{k,1} \begin{pmatrix}
        \cos(u_{k,2}) \\ \sin(u_{k,2})
    \end{pmatrix} + v_k
\] where $v_k \sim \mathcal{N}(0, 0.75I)$. The region of interest is $X = [-10, 10]^2$ and the action space is $[-1, 1]^2$ uniformly discretized into $L^u$ points.
The benchmark is similar to the SySCoRe running example, except with nonlinear inputs.
For running the benchmark using our methods, we transform the input to $v = \begin{pmatrix}u_1\cos(u_2) & u_1\sin(u_2)\end{pmatrix}^\top$ such that the dynamics are linear in the (nonlinearly partitioned) control inputs.

We consider two different specifications: reachability and reach-avoid, each with a different partitioning. For both specifications, the reach region is $R = [5, 7]^2$ and for the reach-avoid specification, the avoid region is $O = [-2, 2]^2$. For reachability, we use a state partitioning of $|\tilde{S}_1| \times |\tilde{S}_2|=(20, 20)$ and an action partitioning of $|A_1| \times |A_2| = (11, 11)$, while for reach-avoid we use $|\tilde{S}_1| \times |\tilde{S}_2| = (40, 40)$ and $|A_1| \times |A_2| = (21, 21)$.
For both properties, we maximize the pessimistic reach-avoid probability over 10 time steps.

\paragraph*{4D building automation system}
A 4D case study from \cite{10.1007/978-3-031-68416-6_15} of building automation system with the linear dynamics $x_{k+1} = Ax_k + B u_k + v_k$ where $v_k \sim \mathcal{N}(0, \Sigma)$ with \[
\begin{aligned}
A &= \begin{pmatrix}
    0.6682 & 0 &  0.02632 &  0 \\
    0 &  0.683 & 0   &  0.02096\\
    1.0005 & 0&   -0.000499 & 0 \\
    0  & 0.8004 & 0  &    0.1996
\end{pmatrix}, \\
B &= \begin{pmatrix}
    0.1320 \\
    0.1402 \\
    0\\
    0.0
\end{pmatrix}, \text{ and } \\
\Sigma &= \diag(1/12.9199, 1/12.9199, 1/2.5826, 1/3.2276).
\end{aligned}
\]
The region of interest is $X = [18.75, 21.25]^2 \times [29.5, 36.5]^2$ and the input space $U = \{17, 18, 19, 20\}$. 
The safe region is $X$. The goal is to maximize the pessimistic probability of safety in 10 time steps, or by duality, minimize the optimistic reachability probability to the complement of the safe set $X^c$. Then the complement probability of that result is maximized pessimistic safety probability.
We consider a partition of the region of interest of $|\tilde{S}_1| \times |\tilde{S}_2| \times |\tilde{S}_3| \times |\tilde{S}_4| = (5, 5, 7, 7)$.

\paragraph*{Van der Pol Oscillator}
A 2D nonlinear case study from \cite{van2023syscore} with the following dynamics: \[
    x_{k+1} = \begin{pmatrix}
        x_{k,1} + x_{k,2} \tau \\
        x_{k,2} + (-x_{k,1} + (1 - x_{k,1})^2 x_{k,2}) \tau + u_k
    \end{pmatrix} + v_k
\] where $\tau = 0.1$ and $v_k \sim \mathcal{N}(0, 0.2I)$.
The region of interest is $X = [-4, 4]^2$ and the action space is $[-1, 1]$ partitioned into 11 points. The goal is to maximize the pessimistic reachability probability to a reach region $R = [-1.4, -0.7] \times [-2.9, -2.0]$ over 10 time steps.
We consider a partition of the region of interest of $|\tilde{S}_1| \times |\tilde{S}_2| = (50, 50)$.

\paragraph*{\Acrlong{nndm}}
As a complex benchmark with non-smooth dynamics, we consider a \gls{nndm} with dynamics $x_{k+1}= f_\theta(x_k, u_k) + v_k$ where $f_\theta : \mathbb{R}^n \times U \to \mathbb{R}^n$ is a neural network. The specific benchmark is a surrogate model for the classic cartpole model \cite{6313077} where the states are $(x, \dot{x}, \theta, \dot{\theta})$, the region of interest is $X = [-a, a]$ with $a = [2, 34.028235, 0.41887903, 34.028235]$, and the input is $U = \{0, 1\}$ (push the cart left and right, respectively). The noise is $v_k \sim \mathcal{N}(0, 0.1I)$.

The model is a neural network with two hidden layers of 256 neurons and ReLU activation function (no activation function on the output). The input to the network is $6$-dimensional since the two discrete actions are one-hot encoded. The model is trained for 10,000 iterations on batches of 100 samples where the samples are obtained by uniformly sampling from $X$. The ADAM optimization algorithm is used with a learning rate of $1\mathrm{e}{-3}$, exponential decay with a multiplicative factor of $0.999$, and the standard mean-squared error is used as the loss function. For reachability analysis of the neural network, CROWN \cite{10.5555/3327345.3327402, 10.5555/3495724.3495820} is used for each region and input. We use $|\tilde{S}_1| \times |\tilde{S}_2| \times |\tilde{S}_3| \times |\tilde{S}_4| = (20, 20, 24, 20)$.

\paragraph*{$n$-D linear model}
To show scalability, we include a linear model with a configurable dimension $n$. The dynamics are $x_{k+1} = Ax_k + v_k$ where $A$ is a Toeplitz matrix
\[A = \underbrace{\begin{pmatrix}
    0.7 & -0.1 & \cdots & 0\\
    0 & 0.7 & \ddots & \vdots \\
    \vdots & \ddots & 0.7 & -0.1 \\
    -0.1 & \cdots & 0 & 0.7
\end{pmatrix}}_n\]
and $v_k \sim \mathcal{N}(0, 0.1I)$. The safe set is $X = [-1, 1]^n$ and the goal is to stay within $X$ for 10 time steps.
We use a partitioning of $\prod_{i=1}^n |\tilde{S}_i| = (\underbrace{8, \ldots, 8}_n)$. Experiments are conducted with $n = 6$ and $n = 7$.

\paragraph*{Dubin's car Gaussian Process}
A 3D case study from \cite{reed2023promises} where the dynamics of each mode (input), of 7, are learned as a \acrfull{gp} with Deep Kernel Learning. That is, a \acrlong{gp} $GP(\mu, k_\theta)$ where $\mu : \mathbb{R}^n \to \mathbb{R}$ is a mean function and $k_\theta : \mathbb{R}^n \times \mathbb{R}^n \to \mathbb{R}_{\geq 0}$ is a deep kernel, with $k_\theta(x, x') = k(g_\theta(x), g_\theta(x'))$ for a given base kernel $k : \mathbb{R}^s \times \mathbb{R}^s \to \mathbb{R}_{\geq 0}$ and a neural network $g_\theta : \mathbb{R}^n \to \mathbb{R}^s$.
We refer to \cite{reed2023promises} for details on the deep kernel including the training process and on the number of samples for the posterior. Bounds on the standard deviation and covariance are obtained using the method from \cite{patane2022adversarial}. We remark that while \cite{reed2023promises} introduces a refinement algorithm, we use the \gls{gp} bounds without refinement, as refining the abstraction results in a non-gridded abstraction, which is incompatible with the proposed method. 

The property to synthesize an optimal strategy for is reach-avoid with $O = [4, 6] \times [0, 1] \times [-0.5, 0.5]$ and $R = [8, 10] \times [0, 1] \times [-0.5, 0.5]$ for 10 time steps. 
We use a region of interest $[0, 10] \times [0, 2] \times [-0.5, 0.5]$ and a partitioning $|\tilde{S}_1| \times |\tilde{S}_2| \times |\tilde{S}_3| = (80, 16, 20)$. 

\paragraph*{Stochastic switched linear system}
We have designed a new benchmark for the class of stochastic switched systems \cite{boukas2007stochastic, ZAMANI2015183} to showcase the functionality on Gaussian mixture models. The dynamics are a mixture of two Gaussians whose mean depends on $x$ and the switching between the modes is governed by a Bernoulli random variable $Z$ with $\probmeas(Z = 0) = 0.7$, \[
    x_{k + 1} \sim \underbrace{\probmeas(Z = 0)}_{\alpha_1} p_1(x_k) + \underbrace{\probmeas(Z = 1)}_{\alpha_2} p_2(x_k)\]
where \[
\begin{aligned}
    p_1(x_k) &= \mathcal{N}\left(\vphantom{\begin{pmatrix}
        0.1 & 0.9\\
        0.8 & 0.2
    \end{pmatrix}x_k,\, \begin{pmatrix}
        0.3^2 & 0\\
        0 & 0.2^2
    \end{pmatrix}}\right.\underbrace{\begin{pmatrix}
        0.1 & 0.9\\
        0.8 & 0.2
    \end{pmatrix}x_k}_{\mu^1(x_k)},\, \underbrace{\begin{pmatrix}
        0.3^2 & 0\\
        0 & 0.2^2
    \end{pmatrix}}_{\Sigma^1}\left.\vphantom{\begin{pmatrix}
        0.1 & 0.9\\
        0.8 & 0.2
    \end{pmatrix}x_k,\, \begin{pmatrix}
        0.3^2 & 0\\
        0 & 0.2^2
    \end{pmatrix}}\right)\\
    p_2(x_k) &= 
    \mathcal{N}\left(\vphantom{\begin{pmatrix}
        0.8 & 0.2\\
        0.1 & 0.9
    \end{pmatrix}x_k,\, \begin{pmatrix}
        0.2^2 & 0\\
        0 & 0.1^2
    \end{pmatrix}}\right.\underbrace{\begin{pmatrix}
        0.8 & 0.2\\
        0.1 & 0.9
    \end{pmatrix}x_k}_{\mu^2(x_k)},\, \underbrace{\begin{pmatrix}
        0.2^2 & 0\\
        0 & 0.1^2
    \end{pmatrix}}_{\Sigma^2}\left.\vphantom{\begin{pmatrix}
        0.8 & 0.2\\
        0.1 & 0.9
    \end{pmatrix}x_k,\, \begin{pmatrix}
        0.2^2 & 0\\
        0 & 0.1^2
    \end{pmatrix}}\right). 
\end{aligned}
\]
The specification asserted is a reach-avoid property with the reach region $R = [1, 2] \times [0, 1]$ and avoid region $O = [-1, 0] \times [-1, 1]$ over $10$ time steps. We use $|\tilde{S}_1| \times |\tilde{S}_2| = (40, 40)$.

\section{Value iteration over odIMDPs vs IMDPs}
Proposition \ref{prop:divide_and_conquer_algorithm_distribution_containment} proves that the divide-and-conquer algorithm in Eqn. \eqref{Eq:DecoupledRobustValueIteration} is at most as conservative as Eqn. \eqref{Eq:RobustValueIteration} over the \gls{imdp} constructed by multiplying the marginal interval bounds. This might however still be unintuitive, hence we show in this section Eqn. \eqref{Eq:DecoupledRobustValueIteration} applied to the example in Fig. \ref{fig:example_margin_joint_ambiguity_sets} against value iteration on the corresponding \gls{imdp}. The result is displayed in Fig. \ref{fig:example_product_partial_structure_vi} where Eqn. \eqref{Eq:DecoupledRobustValueIteration} on the \gls{odimdp} yields a value of 2.16 to a value of 2.06 on the Eqn. \eqref{Eq:RobustValueIteration} on the \gls{imdp}.
That is, a larger and thus tighter lower bound for the \gls{odimdp}.

\begin{figure}
    \centering
    \includegraphics[width=0.8\linewidth, page=2]{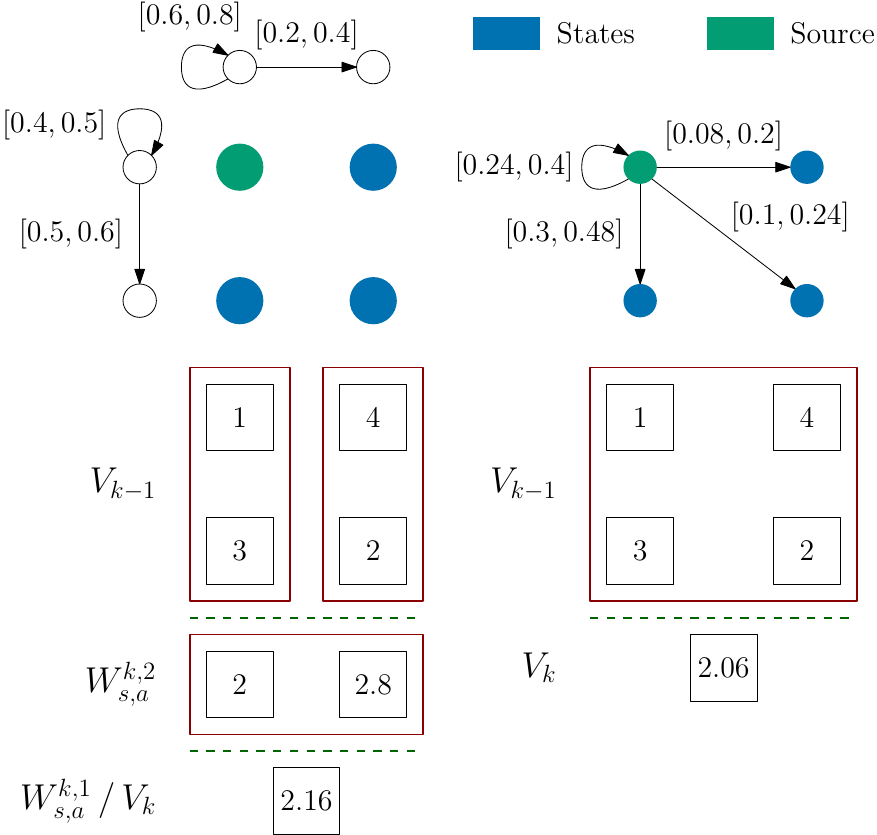}
    \caption{One (pessimistic) Bellman update on the example in Fig. \ref{fig:example_margin_joint_ambiguity_sets} with an \gls{odimdp} on the left and an \gls{imdp} on the right constructed by multiplying the marginal interval bounds of the \gls{odimdp}. The values immediately below the \gls{odimdp} and \gls{imdp} is the value function at the previous step $V_{k-1}$.
    For the \glspl{odimdp}, the Bellman update using Eqn. \eqref{Eq:DecoupledRobustValueIteration} is computed by recursively using O-maximization \cite{givan2000bounded, lahijanian2015formal} with each red box outlining an O-maximization step.
    For the \gls{imdp}, the Bellman update using Eqn. \eqref{Eq:RobustValueIteration} is computed using a one-shot O-maximization step.
    The result is that the \gls{odimdp} and Eqn. \eqref{Eq:DecoupledRobustValueIteration} together yields a tighter lower bound.}
    \label{fig:example_product_partial_structure_vi}
\end{figure}

\end{document}